\documentclass[10pt, a4paper]{article}

\usepackage{graphicx}%
\usepackage{multirow}%
\usepackage{amsmath,amssymb,amsfonts}%
\usepackage{amsthm}%
\usepackage{mathrsfs}%
\usepackage[title]{appendix}%
\usepackage[table, x11names]{xcolor}%
\usepackage{textcomp}%
\usepackage{manyfoot}%
\usepackage{booktabs}%
\usepackage{listings}%

\usepackage{authblk}

\theoremstyle{thmstyleone}%
\newtheorem{theorem}{Theorem}
\newtheorem{proposition}[theorem]{Proposition}%

\theoremstyle{thmstyletwo}%
\newtheorem{remark}{Remark}%

\theoremstyle{thmstylethree}%
\newtheorem{definition}{Definition}%
\newtheorem{assumption}{Assumption}%
\newtheorem{corollary}{Corollary}
\newtheorem{lemma}[theorem]{Lemma}

\usepackage{enumitem}
\usepackage{bm}
\usepackage{amsmath, amssymb, amsthm}
\usepackage{graphicx}
\usepackage{eurosym}
\usepackage{tikz}
\usepackage{xcolor}
\usepackage[ruled,vlined,linesnumbered,noend]{algorithm2e}

\usepackage{subcaption}

\usetikzlibrary{decorations.pathreplacing}

\SetCommentSty{mycommfont}

\usepackage{adjustbox}

\usepackage{multicol}
\usepackage{multirow}

\colorlet{lightgray}{gray!60}
\DeclareMathOperator*{\argmax}{arg\,max}

\newcommand{\minpc}{n}
\newcommand{\np}{s}
\newcommand{\nc}{q}

\newcolumntype{L}[1]{>{\raggedright\let\newline\\\arraybackslash\hspace{0pt}}m{#1}}
\newcolumntype{C}[1]{>{\centering\let\newline\\\arraybackslash\hspace{0pt}}m{#1}}
\newcolumntype{R}[1]{>{\raggedleft\let\newline\\\arraybackslash\hspace{0pt}}m{#1}}

\definecolor{aurometalsaurus}{rgb}{0.21, 0.25, 0.25}

\newcommand{\ind}[1]{{\color{aurometalsaurus}$^{\mathbf{(#1)}}$}}

\newcommand{\ie}{i.e.,}

\newcommand{\greedylocal}[0]{Greedy-Local}
\newcommand{\doublegreedylocal}[0]{Double-Greedy-Local}
\newcommand{\naivealgo}[0]{Naive-Local}
\newcommand{\boundedgreedy}[0]{$\ell$-Greedy-Local}

\newcommand{\naiveedge}[0]{Naive-Edge}

\title{Greediness is not always a vice: Efficient Discovery Algorithms for Assignment Problems}

\author[1]{Romaric Duvignau}
\author[2]{Noël Gillet}
\author[3]{Ralf Klasing}

\affil[1]{Chalmers University of Technology and University of Gothenburg, Sweden}
\affil[ ]{\texttt duvignau@chalmers.se}
\affil[2]{Univ. Orléans, INSA Centre Val de Loire, LIFO UR 4022, FR-45067 Orléans, France}
\affil[ ]{\texttt noel.gillet@univ-orleans.fr}
\affil[3]{CNRS, LaBRI, Universit\'e de Bordeaux, Talence, France}
\affil[ ]{\texttt ralf.klasing@labri.fr}

\begin{document}

\maketitle

\begin{abstract} 
Finding a maximum-weight matching is a classical and well-studied problem in computer science, solvable in cubic time in general graphs.
We consider the specialization called \textit{assignment problem} where the input is a bipartite graph, and introduce in this work the ``discovery'' variant considering edge weights that are not provided as input but must be \textit{queried}, requiring additional and costly computations.
We develop here discovery algorithms aiming to minimize the number of queried weights while providing guarantees on the computed solution.
We first show in this work the inherent challenges of designing discovery algorithms for general assignment problems.
We then provide and analyze several efficient greedy algorithms that can make use of natural assumptions about the order in which the nodes are processed by the algorithms.
Our motivations for exploring this problem stem from finding practical solutions to a variation of maximum-weight matching in bipartite hypergraphs, a problem recently emerging in the formation of peer-to-peer energy sharing communities. 
\end{abstract}

\section{Introduction}

One of the most studied problems in computer science and discrete mathematics, the \textit{assignment problem}, has a very simple formulation, yet a plethora of solutions exist for its many variants and possible additional constraints or optimization aims. 
Using the same nomenclature as used in the rest of the paper, the assignment problem consists in pairing together the members of a first set $P$, often referred to as \textit{Producers} or \textit{agents} in the literature, with members of a second and disjoint set $C$, often referred to as \textit{Consumers} or \textit{tasks}.
The target is to obtain a \textit{One-to-One} correspondence, i.e., each producer can be \textit{assigned} to at most a single consumer and vice-versa. Moreover, as not all producers may be able to serve any particular consumer (and vice-versa), some pairs are considered non valid\footnote{A variant of the problem can set a weight of $0$ for invalid pairs but we rule out such null weights in our formulation. The reason is that our objective is to query as few weights as possible, and weights $0$ are assumed to be already encoded in the input edge set $E \in 2^{P \times C}$.}. For each possible pair $(p,c) \in E$ with $E \subseteq P \times C$, $(p,c)$ is associated with a positive \textit{weight} $w(p,c)$ that represents how much \textit{gain} one can obtain if producer $p \in P$ is paired with consumer $c \in C$.
The assignment problem consists then in finding a \textit{one-to-one assignment} $M \subseteq E$ of the consumers to the producers in order to maximize the total gain $w(M) = \sum_{(p,c) \in M} w(p,c)$, slightly abusing the $w$-notation.
This is a well-studied problem where the Hungarian algorithm~\cite{kuhn1955hungarian} computes an optimal solution in time $\mathcal{O}(n\cdot m + n^2 \log n)$ for $n = \min\{|P|,|C|\}$ and $m = |E|$; see among others~\cite{ramshaw2012minimum} for unbalanced assignment problems and \cite{duan2014linear} for linear-time bounded-approximation algorithms. 
The problem can be alternatively formulated as finding a maximum-weight matching in the bipartite graph $G = (P \cup C, E)$, with the two formulations being equivalent and used interchangeably hereafter for convenience.

A ``discovery'' problem is any optimization problem where the information that is the basis of the optimization is not provided as initial input but must rather be \textit{discovered} during the algorithm's execution.
We extend this notion of discovery algorithms, introduced among others in~\cite{Szepesvari04,CaroDB20}, to assignment problems.
We shall study in this work the \textit{Maximum-Weight Matching Discovery (MWMD) problem} that consists in finding a Maximum-Weight Matching (MWM) using weights that can only be obtained through explicit calls to a computationally-expensive weight function. 
We denote by \textit{query complexity} the number of inspected weights used by a given algorithm to produce its solution.
Since one can easily show that in general, finding the MWM requires the computation of all possible weights in the worst case, we aim to investigate in this paper if approximation algorithms can reach a bounded approximation ratio while requiring the calculation of only an asymptotically subquadratic number of weights in $n$.
Our methods apply to the assignment problems (i.e., bipartite graph matchings) and can be further extended to solve a bipartite version of the hypergraph matching problem with interesting practical applications in energy systems~\cite{duvignau2022efficient,duvignau2024geographical}. 

\paragraph*{Contributions}

\begin{table}[t]
    \caption{Approximation ratios shown in this work for the one-to-one assignment discovery problem over input $G = (P \cup C, E)$, with $n = \min\{|P|,|C|\}$ and $m=|E|$. Each bounded ratio is shown to be achievable (upper bound) by the corresponding algorithm and for each, we show there exist instances (input graph and assumption parameters) where the ratio is reached (lower bound).
    Query complexities are shown in Propositions~\ref{claim:beta_gamma} (Alg.~\ref{alg:matching2}), \ref{prop:alg2_weights} (Alg.~\ref{alg:matching3}) and~\ref{proposition:betal_gammal}
    (Alg.~\ref{alg:double_greedy}), where $\ell \geq 0$ is a parameter of the matching algorithms.
    }
    \label{tab:results}
    \centering
    \footnotesize
    \resizebox{\linewidth}{!}{
    \def\arraystretch{1.45}
    \begin{tabular}{|C{3cm}|| c | c | C{3cm} | c | c | c |}
        \hline
        \rowcolor{gray!30!}

        \textbf{Algorithm} & \textbf{Opt.} & \textbf{Greedy} & \textbf{Alg.~\ref{alg:matching1}} & \textbf{Alg.~\ref{alg:matching2}} & \textbf{Alg.~\ref{alg:matching3}} & \textbf{Alg.~\ref{alg:double_greedy}}   \\
        \hline
        
        \textbf{Query Complexity} & \multicolumn{2}{c|}{$m = |E|$ } & $\leq m$ & $0$ & $\leq (\ell+1) \cdot n$ & $\leq 3 \cdot (\ell+1) \cdot n$  \\
        \cline{1-7}
        
        No weight assumptions & \multirow{10}{*}{1} & \multirow{10}{*}{2} & \multicolumn{4}{c|}{$\infty$ \ind{a}} \\
        \cline{1-1}\cline{4-7}
        $\beta-$strong $P$-order (Assumption~\ref{assumption_beta}) &  &   & $1+\beta$ \ind{b} &  \multicolumn{3}{c|}{$\infty$ \ind{i}}  \\ 
        \cline{1-1}\cline{4-7}
        $\gamma-$strong $C$-order (Assumption~\ref{assumption_gamma}) &  &  & $1+\gamma$ \ind{c} & \multicolumn{3}{c|}{$\infty$ \ind{i}}  \\
        \cline{1-1}\cline{4-7}
        ``Strong orders'' (Assumptions~\ref{assumption_beta} and~\ref{assumption_gamma})  &  &  &$\min \{1 + \beta, \;\; \max\{1, \beta+\gamma\}\}$~\ind{d} & \multicolumn{2}{c|}{$\max \{1, \beta + \gamma\}$ \ind{e}} & $2 \cdot \max \{1, \beta, \gamma \}$ \ind{h}  \\
        \cline{1-1}\cline{4-7}
        Ass.~\ref{assumption_beta} + $\gamma_\ell$-$\ell$-weak $C$-order (Ass. \ref{assumption_gamma_l}) &  &  & $1+\beta$ \ind{b} & $\infty$ \ind{j} & $\beta+\max \{1,\gamma_\ell \}$ \ind{f} & $2 \cdot \max \{1,\beta,\gamma_\ell \}$ \ind{h} \\
        \cline{1-1}\cline{4-7}
        Ass.~\ref{assumption_gamma} +  $\beta_\ell$-$\ell$-weak $P$-order (Ass. \ref{assumption_beta_l}) &  &  & $1+\gamma$ \ind{c} & $\infty$ \ind{j} & $\gamma+\max \{1,\beta_\ell \}$ \ind{g} & $2 \cdot \max \{1,\beta_\ell,\gamma \}$ \ind{h}  \\ 
        \cline{1-1}\cline{4-7}
        ``Weak orders'' (Assumptions~\ref{assumption_gamma_l} and~\ref{assumption_beta_l}) &  &  & \multicolumn{3}{c|}{$\infty$ \ind{k}} & $2 \cdot \max \{1, \beta_\ell, \gamma_\ell \}$ \ind{h}  \\
        
        \hline
    \end{tabular}
   }
   \footnotesize
    \ind{a} Proposition~\ref{claim:oracle_needed};
    \ind{b} Propositions~\ref{theorem_beta} and~\ref{counter_example_beta};
    \ind{c} Remark~\ref{remark:gamma} with Alg.~\ref{alg:matching1} running over input $G=(C \cup P, E)$;
    \ind{d} Remark~\ref{remark:min_matching};
    \ind{e}~Propositions~\ref{claim:beta_gamma},~\ref{prop:counter_example_alg2}, and~\ref{prop:alg3_strong};
    \ind{f} Propositions~\ref{prop_gamma_l} and~\ref{prop:alg3_counter_example};
    \ind{g} Proposition~\ref{prop_beta_l} on $G=(C \cup P, E)$;
    \ind{h} Propositions~\ref{proposition:betal_gammal} and~\ref{example_betal_gammal};
    \ind{i} Remark~\ref{algo:infinity};
    \ind{j}~Remark~\ref{remark:alg2_infinity};
    \ind{k} Propositions~\ref{claim:bounded_search}. 
\end{table}

Recall the greedy matching procedure: consider the edges one by one in decreasing order of weights and add the current edge under consideration whenever both its endpoints are still available at that step of the algorithm. 
It is a folklore result that the greedy matching algorithm produces a $2$-approximate matching $M_g$ compared with the optimal algorithm, i.e., we have $w(M_{opt}) \leq 2 \cdot w(M_g)$ where $M_{opt}$ is the MWM on the input.
Note that both the greedy and the optimal matching (calculated using for instance the Hungarian algorithm) require to inspect the value of all the weights of the input to compute their solution.
The argument for the bounded approximation bound relies on two elements: (1) the order in which the greedy algorithm considers the edges (from largest to smallest weights) and (2) the fact that for each edge $e$ of $M_g$, if $e$ is not present in $M_{opt}$ then it may only be ``replaced'' by two other edges in $M_{opt}$, from which one deduces the approximation bound of $2$.
Our main contribution is to propose a generalization of the above argument to edge sets that are only \textit{partially ordered}, hence allowing to deduce approximation bounds using problem-dependent \textit{heuristic orders} on the vertex sets and this way avoiding to inspect the values of all the weights of the input.
In this work, we introduce the notion of ``order oracles'' (cf. Section~\ref{subsec:oracles}) that are capable to order nodes in specific orders concerning the weights of the edges in their neighborhood without requiring any computation of the edge weights.
This ordering assumption allows us to design efficient greedy algorithms with bounded approximation ratio and requiring to compute only up to $\mathcal{O}(n)$ weights when the vertices of each set are processed in a well-chosen heuristic order.
We summarize our main results in Table~\ref{tab:results}. (``Opt.'' is an optimal matching algorithm, ``Greedy'' refers to the classical greedy algorithm as aforedescribed, while the other algorithms are the ones developed and analyzed in this work. 
Parameters $\beta$, $\gamma$, $\gamma_\ell$, $\beta_\ell$ control the quality of the heuristic orders for processing of the nodes of the input sets $P$ and $C$, and are respectively specified in Assumption~\ref{assumption_beta},~\ref{assumption_gamma},~\ref{assumption_gamma_l} and~\ref{assumption_beta_l}.)

A short and preliminary version of our work appears in~\cite{duvignau2023greediness}. The present work extends~\cite{duvignau2023greediness} and lifts an additional simplifying assumption about the heuristic orders (i.e. that all order parameters $\beta$, $\gamma$, $\gamma_\ell$, $\beta_\ell$ are greater than one), adds a novel and more complex greedy procedure (Alg.~\ref{alg:double_greedy}) and its analysis, achieving a bounded-approximation using only weak orders and a linear number of weight queries, considerations on edge orders and instantiating order oracles as well as further details concerning extending our algorithms to the (bipartite) hypergraph matching problem.

\paragraph*{Motivations} 

In the context of Peer-to-Peer energy sharing~\cite{duvignau2021benefits}, the Geographical Peer Matching (GPM) problem is introduced in~\cite{duvignau2022efficient} to efficiently compute a matching of the peers targeting the maximization of a global objective (i.e., the total cost-savings).
It relies on both geographical information about the peers as well as their local matching preferences, and seeks to build an assignment of the peers into groups of size up to $k$ as advocated by the application.
Building on the discovery algorithms presented and analyzed in this work, we can obtain bounded-approximation algorithms for the GPM problem that run in linear time 
and use only a linear number 
of weight calculations, under certain assumptions occurring in practice (see \S~\ref{one-to-many-assignments:hypergraphs}).

\paragraph*{Related Work} 

Discovery algorithms have been studied in the literature for various problems on weighted graphs. However, as far as we are aware, they have not been investigated so far for the maximum-weight matching problem. 
For any optimization problem (a.k.a. maximization or minimization problems), considering that the solution of the discovery-variant of a given problem (i.e., assuming part of the input is obtained on the fly) is also a valid solution to the original problem where all inputs are provided at the start of the algorithm, the time complexity required to reach an optimal solution is always at least as large as the one for the original non-discovery problem.

Szepesvari \cite{Szepesvari04} introduced the {\em Shortest Path Discovery Problem} (SPDP), in which the task is to discover in a given edge-weighted graph a shortest path for fixed source and target nodes. An algorithm is proposed that is shown to use a small number of queries.
Experimental results on real-world instances are also presented.
Caro et al.~\cite{CaroDB20} generalize the SPDP to the {\em Optimal Path Discovery Problem}. First, they consider a broader class
of cost functions, and relax the constraint that an optimal path has to be discovered, allowing the discovered path to be an $\alpha$-approximation. Second, whereas in \cite{Szepesvari04} the performance of algorithms was measured with the number of queries, Caro et al.~\cite{CaroDB20} propose a more fine-grained performance measure, called the {\em query ratio}, i.e., the ratio between the number of queried edges and the least number of edge values required to solve the problem. They prove a $1+4/n-8/n^2$ lower bound on the query ratio and present an algorithm whose query ratio, when it finds the optimal path, is upper bounded by $2-1/(n-1)$, where $n = |V|$. Finally, they implement different algorithms and evaluate their query ratio experimentally.

Erlebach et al.~\cite{ErlebachHKMR08} consider the minimum spanning tree problem with {\em queryable uncertainty}. This concept refers to settings where the input of a problem is initially not known precisely, but exact information about the input can be obtained at a cost using queries. An algorithm with query ratio 2 is proposed in \cite{ErlebachHKMR08} for the minimum spanning tree problem, and it is shown that this query ratio is the best possible among deterministic algorithms. In \cite{ErlebachHK16}, the authors extend the framework to cheapest set problems with queryable uncertainty that englobe previously studied problems such as the minimum spanning tree, or the minimum matroid base problem under queryable uncertainty. For the cheapest set problems with queryable uncertainty, the authors present an algorithm that makes $d \cdot {\rm OPT} + d$ queries, where OPT is the optimal number of queries required to solve the problem and $d$ is the maximum cardinality of a feasible set in a given instance. An algorithm with query ratio 2 for the minimum matroid base problem is also provided in \cite{ErlebachHK16}. In~\cite{ErlebachHL23}, algorithms for uncertainty problems are studied in which parallel queries are allowed. Round-competitive algorithms are presented for sorting, selection, and for the minimum value problem. In \cite{Erlebach21}, a survey on models and algorithms for problems that access the input via queries can be found.

Another similar line of work considers the robust spanning tree problem with interval data. For a given graph with weight intervals specified for its edges, the goal is to compute e.g.~a spanning tree that minimizes the worst-case deviation from the minimum spanning tree (also called the {\em regret}), over all realizations of the edge weights. This is an off-line problem, and no query operations are involved. The problem is proved NP-hard in \cite{aron2004complexity} while a 2-approximation algorithm is given in \cite{kasperski2006approximation}. Further work has considered heuristics or exact algorithms for the problem, see e.g.~\cite{yaman2001robust}.

Regret minimization was also considered for other combinatorial optimization problems with interval data. Indeed, for problems in P (including the \textsc{Assignment} problem) there is a generic method to obtain constant approximations with respect to the regret~\cite{kasperski2006approximation}.
On the contrary, this was shown not to be true in general for NP-hard optimization problem, by Ganesh et al.~\cite{GaneshMP23}.
For that reason they developed novel techniques for regret minimization of NP-hard optimization problems, opening the door for a new and exciting research direction. 
The result is the first constant factor approximation algorithm for the robust setting of NP-hard optimization problems, including the classical problems \textsc{TSP} 
on metric graphs and \textsc{Steiner Tree}.

The \emph{network verification} problem is that of establishing the accuracy of a high-level description of its physical topology, by making as few measurements as possible on its nodes. This task can be formalized as a {\em Network Discovery} optimization problem that, given a graph and a query model specifying the information returned by a query at a node, asks for finding a minimum-size subset of nodes to be queried so as to univocally identify the graph. This problem has been studied with respect to different query models, assuming that a node has some global knowledge about the network~\cite{BampasBDGKP15,BeerliovaEEHHMR06,BiloEMW10,ErlebachHM07}.

\paragraph*{Plan} 
In Section~\ref{sec:oracles}, we define the assignment discovery problem, show its inherent challenges and hence the need for introducing order oracles to analyze the performance of discovery algorithms. 
In Section~\ref{sec:discovery_algorithms}, we present several greedy algorithms producing a matching without querying the totality of the weights, and analyze them relying on different assumptions about the order in which the nodes are processed in regard to the weights of the edges.
We further complement the section considering orders on edges and how to instantiate order oracles using interval weights or an approximation function in lieu of precise weights.
In Section~\ref{sec:hypergraph_matchings}, we present how our algorithms extend to one-to-many assignment problems, before concluding our work in Section~\ref{sec:conclusions}.

\section{Order Oracles for the Assignment Discovery Problem} \label{sec:oracles}

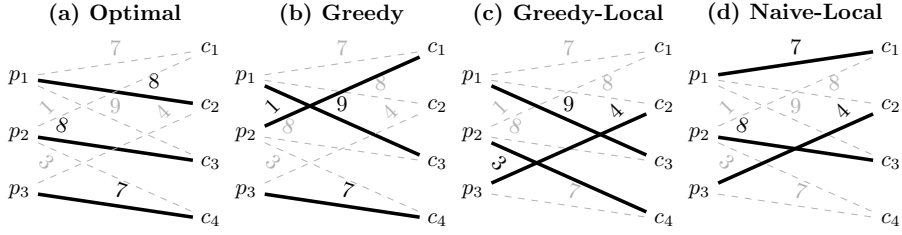
\begin{figure*}[t]
      \centering
      \resizebox{\linewidth}{!}{
\begin{tikzpicture}[yscale=0.45]
        \node (p1) at (0,0) {$p_1$};
        \node (c1) at (3,1) {$c_1$};
        \node (c2) at (3,-1) {$c_2$};
        \node (p2) at (0,-2) {$p_2$};
        \node (c3) at (3,-3) {$c_3$};
        \node (p3) at (0,-4) {$p_3$};
        \node (c4) at (3,-5) {$c_4$};
        \draw[-, dashed, lightgray] (p1) -- node[midway, above] {$7$} (c1);
        \draw[-, ultra thick] (p1) -- node[near end, above,font=\bfseries]  {$8$} (c2);
        \draw[-, dashed, lightgray] (p1) -- node[midway, above]  {$9$} (c3);
         \draw[-, dashed, lightgray] (p2) -- node[above, very near start, sloped] {$1$} (c1);
        \draw[-, ultra thick] (p2) -- node[above, very near start, sloped,font=\bfseries] {$8$} (c3);
         \draw[-, dashed, lightgray] (p2) -- node[below, very near start, sloped] {$3$} (c4);

         \draw[-, dashed, lightgray] (p3) -- node[above, very near end, sloped] {$4$} (c2);
         \draw[-, ultra thick] (p3) -- node[above, midway, sloped,font=\bfseries] {$7$} (c4);

         \node[above,font=\bfseries] at (current bounding box.north) {(a) Optimal};
    \end{tikzpicture}%
    \hfill%
    \begin{tikzpicture}[yscale=0.45]
        \node (p1) at (0,0) {$p_1$};
        \node (c1) at (3,1) {$c_1$};
        \node (c2) at (3,-1) {$c_2$};
        \node (p2) at (0,-2) {$p_2$};
        \node (c3) at (3,-3) {$c_3$};
        \node (p3) at (0,-4) {$p_3$};
        \node (c4) at (3,-5) {$c_4$};
        \draw[-, dashed, lightgray] (p1) -- node[midway, above] {$7$} (c1);
        \draw[-, dashed, lightgray] (p1) -- node[near end, above]  {$8$} (c2);
        \draw[-, ultra thick] (p1) -- node[midway, above,font=\bfseries]  {$9$} (c3);

         \draw[-, ultra thick] (p2) -- node[above, very near start, sloped, font=\bfseries] {$1$} (c1);
        \draw[-, dashed, lightgray] (p2) -- node[above, very near start, sloped] {$8$} (c3);
         \draw[-, dashed, lightgray] (p2) -- node[below, very near start, sloped] {$3$} (c4);

         \draw[-, dashed, lightgray] (p3) -- node[above, very near end, sloped] {$4$} (c2);
         \draw[-, ultra thick] (p3) -- node[above, midway, sloped,font=\bfseries] {$7$} (c4);

         \node[above,font=\bfseries] at (current bounding box.north) {(b) Greedy};
    \end{tikzpicture}%
    \hfill%
\begin{tikzpicture}[yscale=0.45]
        \node (p1) at (0,0) {$p_1$};
        \node (c1) at (3,1) {$c_1$};
        \node (c2) at (3,-1) {$c_2$};
        \node (p2) at (0,-2) {$p_2$};
        \node (c3) at (3,-3) {$c_3$};
        \node (p3) at (0,-4) {$p_3$};
        \node (c4) at (3,-5) {$c_4$};
        \draw[-, dashed, lightgray] (p1) -- node[midway, above] {$7$} (c1);
        \draw[-, dashed, lightgray] (p1) -- node[near end, above]  {$8$} (c2);
        \draw[-, ultra thick] (p1) -- node[midway, above,font=\bfseries]  {$9$} (c3);
        
         \draw[-, dashed, lightgray] (p2) -- node[above, very near start, sloped] {$1$} (c1);
        \draw[-, dashed, lightgray] (p2) -- node[above, very near start, sloped] {$8$} (c3);
         \draw[-, ultra thick] (p2) -- node[below, very near start, sloped,font=\bfseries] {$3$} (c4);

         \draw[-, ultra thick] (p3) -- node[above, very near end, sloped,font=\bfseries] {$4$} (c2);
         \draw[-, dashed, lightgray] (p3) -- node[above, midway, sloped] {$7$} (c4);

         \node[above,font=\bfseries] at (current bounding box.north) {(c) \greedylocal{}};
    \end{tikzpicture}%
    \hfill%
    \begin{tikzpicture}[yscale=0.45]
        \node (p1) at (0,0) {$p_1$};
        \node (c1) at (3,1) {$c_1$};
        \node (c2) at (3,-1) {$c_2$};
        \node (p2) at (0,-2) {$p_2$};
        \node (c3) at (3,-3) {$c_3$};
        \node (p3) at (0,-4) {$p_3$};
        \node (c4) at (3,-5) {$c_4$};
        \draw[-, ultra thick] (p1) -- node[midway, above,font=\bfseries] {$7$} (c1);
        \draw[-, dashed, lightgray] (p1) -- node[near end, above]  {$8$} (c2);
        \draw[-, dashed, lightgray] (p1) -- node[midway, above]  {$9$} (c3);
        
         \draw[-, dashed, lightgray] (p2) -- node[above, very near start, sloped] {$1$} (c1);
        \draw[-, ultra thick] (p2) -- node[above, very near start, sloped,font=\bfseries] {$8$} (c3);
         \draw[-, dashed, lightgray] (p2) -- node[below, very near start, sloped] {$3$} (c4);

         \draw[-, ultra thick] (p3) -- node[above, very near end, sloped,font=\bfseries] {$4$} (c2);
         \draw[-, dashed, lightgray] (p3) -- node[above, midway, sloped] {$7$} (c4);

         \node[above,font=\bfseries] at (current bounding box.north) {(d) \naivealgo{}};
    \end{tikzpicture}
    }
\caption{Examples of matchings: (a) Optimal with weight $23$, (b) Greedy with weight $17$, (c) \greedylocal{} (Alg.~\ref{alg:matching1}) with weight $16$, (d) \naivealgo{} (Alg.~\ref{alg:matching2}) with weight $19$; the $1$-\greedylocal{} (Alg.~\ref{alg:matching3}) algorithm outputs the matching (a) as well as the $1$-\doublegreedylocal{} (Alg.~\ref{alg:double_greedy}). Here, the strong and weak ordering assumptions hold with $\beta = 7/3$, $\gamma = 8$, $\beta_1 = 0$, $\gamma_1=3$ and $\gamma_2=0$.}
    \label{fig:examples}
\end{figure*}

\subsection{Preliminaries}


\begin{table*}[b]
    \caption{Symbols used throughout the paper.} 
    \label{tab:nomenclature}
    \centering
    \fbox{
    \resizebox{\textwidth}{!}{%
    \begin{tabular}{l l | l l}
        \textbf{Symbol} & \textbf{Definition} & \textbf{Symbol} & \textbf{Definition} \\
        \hline
        $G$ & input weighted graph $G =(P\cup C, E, w)$ & $V$ & a vertex set, $V = P \cup C$  \\ 
        $P$ & a set of “agents” & $s$ & size of $P$ \\
        $C$ & a set of “tasks” & $q$ & size of $C$ \\
        $E$ & a set of edges (allowed agent-task pairs) &  $m$ & size of $E$ \\
        $M$ & a matching of the edges & $n$ & maximum matching size, $n = \min\{s,q\}$ \\
        $w(e)$ & weight associated with edge $e \in E$ & $w(M)$ & weight of the matching $M$ \\ 
        $\sigma_P$ & an order over (or permutation of) $P$ &  $\sigma_C$ & an order over (or permutation of) $C$ \\
        $\mathcal{A}$ & a matching algorithm &  $\ell$ & a parameter of some matching algorithms \\
        $\beta, \beta_\ell$ & parameters associated with $\sigma_P$ &  $\gamma, \gamma_\ell$ & parameters associated with $\sigma_C$ \\
    \end{tabular}
    }
    }
\end{table*}

We adopt the following conventions for the notation used hereafter. Let $G = (P \cup C, E)$ denote a bipartite graph serving as our input instance; $P$, a set of ``agents'' to match with ``tasks'' with $\np = |P|$ the number of considered agents; $C$, a set of tasks with $\nc = |C|$;  $E \subseteq P \times C$ the set of possible edges with $(p,c) \in E$ if the task $c$ can be assigned to the agent $p$ and $m = |E|$; $w(e) \in \mathbb{R}^+$ for $e\in E$ is the weight of the edge $e$; $\minpc = \min \{|P|,|C|\}$ is the maximum size of a matching in $G$. 
We slightly abuse the $w$ notation so as to write $w(p,c)$ to denote as well the weight of the edge $(p,c)$ and $w(M)$ for the weight of the matching $M \subseteq E$, i.e.,  $w(M) = \sum_{e\in M} w(e)$.
Refer to Table~\ref{tab:nomenclature} for a quick reference to the definition of the symbols used throughout the paper.

An isolated edge is any edge without any adjacent edges in $G$, i.e., $e = (p,c)$ is isolated if $\lnot (\exists (p,c') \in E, c' \not= c \lor \exists (p',c) \in E, p' \not= p)$.
In the following, we assume that all weights are strictly positive as edges with negative or zero weight are assumed to be removed from the considered input graph.
The {\em query complexity} of an algorithm $\mathcal{A}$ is the number of weights $\mathcal{A}$ \textit{inspects}\footnote{We use hereafter interchangeably the terms \textit{inspect}, \textit{discover}, \textit{query} and \textit{compute} for the same action of checking the value of the weight $w(e)$ of one of the input edges $e \in E$. Because such weight can easily be memorized by the algorithm, we only account for the first inspection of $w(e)$.} in the {\em worst-case} in order to calculate its output. 
Examples of matchings are provided in Figure~\ref{fig:examples} with the discovery algorithms computing them being defined in Section~\ref{sec:discovery_algorithms}.

\begin{definition}
    For $\alpha \geq 1$, we refer for a given matching algorithm $\mathcal{A}$ as being {\em $\alpha$-approximate} if and only if for all possible inputs $G \in \mathcal{G}$ of $\mathcal{A}$, the output matching of  $\mathcal{A}$ denoted $M_{\mathcal{A}}(G)$ has weight at least $1/\alpha$ of the optimal matching $M_{opt}(G)$ of $G$, i.e., $\forall G \in \mathcal{G}, w(M_{opt}(G)) \leq \alpha \cdot M_{\mathcal{A}}(G)$.
\end{definition}

We have chosen a weighting function $w$ taking values in $\mathbb{R}^+$, however, all our results can be shown to also apply when $w$ is restricted to integer weights as our arguments only rely on weight orders and bounds rather than the actual values. Hence, if not explicitly stated, $w$ can be restricted to take only integer values.
To be more precise, all our algorithms work fine with integer weights but we have used rational weights for some graph instances within the lower bound arguments, hence a simple scaling of all weights will entail an argument that is valid for integer weights as well.

\subsection{The Need for Order Oracles} \label{subsec:oracles}

\paragraph{Lower bounds for the number of discovered edges}

We first show that, without additional assumptions, any algorithm requires in the worst case the computation of all possible weights in $G$ (discarding isolated edges) in order to reach a bounded approximation ratio. Note if $G$ has isolated edges, all such edges can be added safely without computing their weights and we can assume that the considered algorithms rather start with $G'$, the graph obtained by stripping all isolated edges from $G$.
Let us first observe the following simple result that restrains the edges that are not inspected when producing a matching for any input graph (not necessarily bipartite).

\begin{lemma} \label{claim:oracle_needed_lemma} 
 Let $\alpha \geq 1$. For any input graph $G$, any $\alpha$-approximate matching discovery algorithm $\mathcal{A}$ must include in its output matching all the edges of $G$ whose weight are never inspected by $\mathcal{A}$.
\end{lemma}

\begin{proof} 
Let $\mathcal{A}$ be an $\alpha$-approximate matching discovery algorithm, $G = (V, E, w)$ be an input graph, and $M$ the matching computed by $\mathcal{A}$ over $G$.
Suppose there exists $e \in E$, $e \not \in M$ and $e$'s weight is never queried by $\mathcal{A}$ while calculating $M$.
Since $w(e)$ is not queried by $\mathcal{A}$, it can be arbitrarily large, as for example $w(e) = \alpha' \sum_{e' \in E \setminus \{e\}} w(e')$ with $\alpha' > \alpha$.
Hence, any matching of $G$ including $e$ is at least $\alpha'$ better than any matching not including it, implying $\mathcal{A}$ is not $\alpha$-approximate in this case.
\end{proof}

One may notice that if a matching discovery algorithm $\mathcal{A}$ greedily adds an edge $e$ to the output matching (after inspecting its weight or not), then the weights of all the \textit{adjacent} edges to $e$ (i.e. those edges sharing an endpoint with $e$) must have been inspected before adding the edge $e$ to the matching. This is due to the impossibility for $\mathcal{A}$ to add them later to the matching due to its greedy decision concerning $e$ and the previous lemma (an edge whose weight is unknown must appear in the final matching).
Furthermore, we can also deduce from the previous lemma the following result.

\begin{lemma} \label{claim:oracle_needed_lemma2}
For any input graph $G = (V, E, w)$ and $\alpha \geq 1$, any $\alpha$-approximate matching discovery algorithm $\mathcal{A}$ examines at least $|E|-\lfloor\frac{|V|}{2}\rfloor$ edges to compute its solution.
\end{lemma}

\begin{proof} 
First note that for any matching $M$ of the edges of $G$, any edge $e \in M$ of the matching necessarily \textit{blocks} two vertices of $V$ from being used in the other edges of $M \setminus \{e\}$, thus $|M| \leq \lfloor\frac{|V|}{2}\rfloor$.
Now, applying Lemma~\ref{claim:oracle_needed_lemma}, any edge that is not inspected by $\mathcal{A}$ must also be included in the output matching $M_\mathcal{A}(G)$, hence at most $\lfloor\frac{|V|}{2}\rfloor$ edges are not inspected.

\end{proof}

\begin{corollary}
\label{cor:query_lower}
    If $|E| = \Omega(|V|^2)$, any bounded approximation algorithm must query $\Omega(|V|^2)$ edges of the input graph $G = (V, E, w)$.
\end{corollary}

For bipartite graphs of the form $G = (P \cup C, E)$, the lower bound on the number of queried edges may be slightly higher than in Lemma~\ref{claim:oracle_needed_lemma2}, as any edge of the matching eliminates both a node in $P$ and one in $C$, entailing $|M| \leq \min\{|P|,|C|\}$ and thus $m-n$ edges must be queried by any bounded approximation matching discovery algorithm.  
Since our target is to query at most a linear number of weights in $n$, there exists no such efficient discovery algorithm for general (bipartite) graphs. We show hereafter an even tighter lower bound on $m$ for some graph families. 
We clarify that the following result appears already in~\cite{duvignau2023greediness} but the proof arguments of~\cite{duvignau2023greediness} are only explicit when $G$ contains $2$ edges.

\begin{proposition} \label{claim:oracle_needed}
There exist unbounded graph families that do not admit a bounded-approximation algorithm $\mathcal{A}$ for maximum weight matching such that $\mathcal{A}$ queries strictly less than $m = |E|$ weights, for the input graph $G = (V,E,w)$.
\end{proposition}

\begin{proof} 
We give the proof for $G$ being a star of $m \geq 2$ edges, i.e., $G = (P \cup C,E)$ with $P = \{p\}$ and $C = \{c_1, \dots, c_{\nc}\}$ such that $E = \{(p,c_j) \;|\; 1 \leq j \leq q\}$.

Suppose there exists an $\alpha$-approximate algorithm $\mathcal{A}$ that always avoids the computation of at least one weight $w(p,c) > 0$ for a non-isolated edge $(p,c)$ in a given input graph $G_w$, i.e. the star graph $G$ equipped with the weight function $w$.
Let $w(p,c_1) = \alpha'$ with $\alpha' > \alpha$ and $w(p,c_j) = 1$ for $2 \leq j \leq q-1$ such that those edges also match in the same order the edges whose weight is queried by $\mathcal{A}$ (because of symmetries in the star-graph, this order is only dependent on $\mathcal{A}$). Hence,  the edge $(p,c_{\nc})$ is the edge that is never discovered by $\mathcal{A}$.
By Lemma~\ref{claim:oracle_needed_lemma}, we know that $\mathcal{A}$ selects the edge $(p,c_{\nc})$ in its output. 
Hence, if e.g. $w(p,c_{\nc}) = 1$, the matching produced by $\mathcal{A}$ is less than $\alpha$ times the optimal, that selects $(p,c_1)$ here for a total weight of $\alpha' > \alpha$.

\end{proof}

\paragraph{Avoiding weight calculations}

The above proposition claims that there exist arbitrarily large instances where a bounded approximation algorithm has no other choice than inspecting the weight of all edges in the input graph. However, there also exist instances where a bounded approximation can be obtained without checking the weights of all non-isolated edges. Consider a path $P_4$ made of $4$ connected edges $e_1, e_2, e_3, e_4$ and the following algorithm. If $w(e_2) > w(e_3)$ then return $\{e_4, \argmax_{e \in \{e_1, e_2\} } w(e) \}$ else return $\{e_1, \argmax_{e \in \{e_3, e_4\} } w(e) \}$. The algorithm always skips the calculation of one of the weights, however, it always produce a solution that is $2$-approximate (simply because in each case, it has already accumulated at least half of the optimal without accounting for the non-queried edge).
Using a disjoint union of $P_4$ as input, one can show that it is possible to avoid at least $m/4 = \Omega(|E|)$ weight calculations in some graph instances. 

\paragraph{Order oracles} 

Proposition~\ref{claim:oracle_needed} essentially tells us that without additional assumptions, one may need to compute all $m'$ weights (where $m'$ is the number of non-isolated edges in the input). In this scenario, one can simply run the optimal algorithm on $G'$ and add all isolated edges afterwards which obviously produces the optimal solution for $G$. 
To circumvent the impossibility and aim to compute less than $m'$ weights, we assume that there exists an oracle that provides us with the vertices of $P$ and possibly of $C$ in an order $\sigma_P$ (or $\sigma_C$) which guarantees additional properties about the weights. The matching algorithm $\mathcal{A}$'s aim is to heuristically use $\sigma_P$ and $\sigma_C$ to avoid to query the weight of some of the edges.
More generally, one may assume that the oracle is powerful enough to provide the edges that the matching algorithm should consider in an order $\sigma_E$ over $E$ so that edges with higher weights are generally considered earlier on.
In such a case, observe that for any given matching algorithm $\mathcal{A}$, there exists an optimal order $\sigma_{\mathcal{A}}$ over $E$ that optimizes the weight of the matching produced by $\mathcal{A}$ (note, for some algorithms, all such orders may still produce the same result). 
Since our goal is to design efficient matching algorithms that minimize the number of weight calculations, we cannot assume that edges are processed by $\mathcal{A}$ in the order $\sigma_{\mathcal{A}}$ but rather the goal is to design an algorithm $\mathcal{A'}(\sigma_P,\sigma_C)$ which produces a matching of bounded approximation ratio given the oracle's orders, while calculating a hopefully limited number of weights.
Assuming there exist heuristic orders on $P$ and $C$ with interesting properties on the weight function stems from the settings of our original motivating problems of peer matching among energy communities~\cite{duvignau2022efficient}.
In the next section, we design greedy algorithms exploiting $\sigma_P$ and $\sigma_C$ and show their approximation ratio.
Our aim is to assume that $\sigma_P$ and $\sigma_C$ entail weak properties on the weights but strong enough to be able to reach a sub-quadratic number of weight calculations in $\minpc$ while keeping a bounded approximation ratio for the calculated matching.
Observe that order oracles that provide us with a total ordering of the edges are very strong, cf. \S~\ref{sec:edge_orders}.

\section{Discovery Algorithms for the One-to-One Assignment Problem} \label{sec:discovery_algorithms}

\subsection{Order Oracles for the Vertex Sets}

\subsubsection{The \greedylocal{} Algorithm}

\begin{figure*}[t]
      \centering
\begin{minipage}{0.47\linewidth}
\begin{algorithm}[H]
\footnotesize
\SetAlgoLined
\SetKwInOut{Input}{Input}
\SetKwInOut{Output}{Output}
\SetKwInOut{Upon}{Event}
\vspace{0.05cm}
\Input{
A bipartite graph $G = (P \cup C, E)$ with sets $P = p_1, p_2, \dots, p_{\np}$ and $C = c_1,c_2, \ldots, c_{\nc}$
}
\Output{$M$, a matching of $E$\;}
\vspace{0.05cm}
    \tcp{Initialization}
    $M \leftarrow \emptyset$ \; \label{alg1:line1}
    
    \ForEach{$j \in C$}{ 
        $\text{available}_j \leftarrow $ True \; \label{alg1:line3}
    }

    \tcp{Greedy Matching Loop}
    \For{$1 \leq i \leq \np$}{ \label{alg:loopstart}
            $N \leftarrow \{ 1 \leq j \leq \nc  \;|\; \{p_i,c_j\} \in E \land  \text{available}_j \}$ \;
            
            \If{$N \not= \emptyset$}
            {
                \eIf{$|N| > 1$}{
                    \ForEach{$j \in N$}{
                        $b_j \leftarrow$ \textsf{weight}$(p_i, c_j)$\;
                    }
                    $j \leftarrow \argmax_{j \in N} \; b_{j}$; $\quad$ \tcp{In the case of a tie, take the smallest index $j$} \label{line:argmax}
                }
                {
                    $j \leftarrow N[1]$;  \tcp{Retrieve the first value}
                }
                
                $\text{available}_j \leftarrow$ False \;
                $M \leftarrow M \cup \{p_i,c_j\}$\;
            }
        }
\Return $M$\;
\caption{\greedylocal{} Matching}
\label{alg:matching1}
\end{algorithm}
\end{minipage}%
\hfill%
\begin{minipage}{0.52\linewidth}
\begin{algorithm}[H]
\footnotesize
\SetAlgoLined
\SetKwInOut{Input}{Input/Output}

    \For{$1 \leq i \leq \np$}{ \label{alg:loopstart}
            $N \leftarrow \{ 1 \leq j \leq \nc  \;|\; \{p_i,c_j\} \in E \;\land\;  \text{available}_j\}$\;
            \If{$N \not= \emptyset$}
            {
                
                $j \leftarrow N[1]$ \;
                $\text{available}_j \leftarrow$ False\;
                $M \leftarrow M \cup \{p_i,c_j\}$\;
            }
        }
\Return $M$\;
\caption{\naivealgo{} Matching}
\label{alg:matching2}
\end{algorithm}
\vspace{-0.17cm}
\begin{algorithm}[H]
\footnotesize
\SetAlgoLined
\SetKwInOut{Input}{Input/Output}

    \For{$1 \leq i \leq \np$}{ \label{alg:loopstart}
            $N \leftarrow \{ 1 \leq j \leq \nc  \;|\; \{p_i,c_j\} \in E \land  \text{available}_j \}$ \;
            \If{$N \not= \emptyset$}
            {
                \eIf{$|N| > 1$}{
                    \tcp{Keep the $\ell+1$ first values}
                    $N \leftarrow N[:\ell+1]$ \;
                    \ForEach{$j \in N$}{
                        $b_j \leftarrow$ \textsf{weight}$(p_i, c_j)$\;
                    }
                    
                    $j \leftarrow \argmax_{j \in N} \; b_{j}$; \tcp{As in Alg.~\ref{alg:matching1}, line~\ref{line:argmax}}
                }
                {
                    $j \leftarrow N[1]$ \;
                }
                
                $\text{available}_j \leftarrow$ False ; $M \leftarrow M \cup \{p_i,c_j\}$\;
            }
        }
\Return $M$\;
\caption{\boundedgreedy{} Matching}
\label{alg:matching3}
\end{algorithm}
\end{minipage}
\begin{center}
\footnotesize All $3$ algorithms have the same input/output (as Alg.~\ref{alg:matching1}) and Alg.~\ref{alg:matching2} and Alg.~\ref{alg:matching3} start by the same initialization lines (\ref{alg1:line1}-\ref{alg1:line3}) as in Alg.~\ref{alg:matching1}.
\end{center}
\end{figure*}

Before studying more efficient discovery algorithms and to introduce important ordering assumptions and proof arguments, we first study the following simple greedy procedure Alg.~\ref{alg:matching1}: the vertices of the set $P$ are processed one by one in the oracle's order $\sigma_P$ where $\sigma_P$ was designed to have earlier vertices more likely to be associated with higher gains for a given task than vertices appearing later in the order. 
Each time a node is processed, its full neighborhood is examined and the available edge with highest weight is selected to be added to the matching.
In the following, during the round where $p\in P$ is considered, we refer to an edge $(p,c)$ as being \textit{available} if the endpoint $c \in C$ of the edge has not been previously \textit{blocked} by adding another edge $(p',c)$ to the matching at an earlier stage of the algorithm (which is greedy and never reconsiders previous choices).

We show that this \textit{\greedylocal{}} matching algorithm achieves a bounded approximation ratio if $\sigma_P = p_1, \dots, p_{\np}$ orders the vertices in $P$ such that for any $1 \leq i < j \leq \np$, the weight of $(p_j,c)$ is upper-bounded by $\beta$ times the weight of $(p_i,c)$, for any $c$ such that both $(p_i,c) \in E$ and $(p_j,c) \in E$. 
In Proposition~\ref{theorem_beta}, we show that Alg.~\ref{alg:matching1} produces a $(1+\beta)$-approximate matching under the aforementioned ordering assumption (Assumption~\ref{assumption_beta}, referred to in the following as ``$\beta-$strong $P$-order'').
Note that if $\beta \geq 1$, the approximation bound is weaker than the classical greedy matching which is $2$-approximate.
Also, whenever weights of the input graph may be equal to each other and $\beta \not= 0$, the ``best value'' that $\beta$ may take is $1$ (i.e.~$\beta \geq 1$ because any subsequent edge sharing an endpoint in $C$ with an edge being processed may have an equal or strictly smaller weight).
Observe that without any ordering assumptions, Alg.~\ref{alg:matching1} does not produce a bounded approximation in general as its greedy decisions do not take the ``future'' into consideration, hence adding the edge $(p_i,c)$ to the matching might remove the possibility to add a later-to-be-processed edge $(p_j,c)$, with $j > i$, and whose weight might be arbitrarily large.

\begin{assumption} \label{assumption_beta}
($\beta-$strong $P$-order)
Assume that $\beta \geq 0$ and $P$ is processed in the order $\sigma_P = p_1, p_2, \dots, p_{\np}$, so that for any $p_i,p_j \in P$ with $1 \leq i < j \leq \np$ and $c \in C$ such that $(p_i,c) \in E$ and $(p_j,c) \in E$, we have $ w(p_j, c) \leq \beta \cdot  w(p_i, c)$. 
\end{assumption}

We can remark here that assuming a $P$-order is a weaker assumption than in classical greedy ordering, in the sense that, it does not require a total order over all the edges of $G$. Indeed, the property is only \emph{local} to each node $c \in C$, for which we can bound the error of adding an early $(p_i,c)$ edge in the matching, without requesting the weight of the next $(p_{j},c)$ edges for $j>i$. All the considered orders in this work are thus only \textit{partial} edge orders, see also Remark~\ref{remark:comparable_edges}.

\begin{proposition} \label{theorem_beta}
Under $\beta$-strong $P$-order, Alg.~\ref{alg:matching1} has approximation ratio at most $1 + \beta$.
\end{proposition}

\begin{proof}
Let $M$ be the matching obtained by an optimal algorithm and $M'$ the one by Alg.~\ref{alg:matching1}.
The main idea behind the proof is based on the fact that if an edge $e$ is present in an optimal matching $M$ but not in the matching $M'$ computed by our algorithm, it implies that there is at least one adjacent edge $e'\in M'$ that \emph{blocks} $e$ from being selected into $M'$. 
We further demonstrate that there are at most two such blocking edges for any non selected edge of $M$. 

Let $f: M \rightarrow M'$ be a function that projects the edges selected by the matching $M$ onto the edges of $M'$ defined as follows:
\begin{itemize}
    \item[(1)] For $e \in M$, if $e \in M'$, then $f(e) = e$.
    \item[(2)] For $e = (p_j,c) \in M$ and $e \not\in M'$, consider the two following cases.
    \begin{itemize}
    \item[(a)] At the beginning of $p_j$'s turn, $e$ was not selected in $M'$ because it was already \textit{blocked}. That is, $e$ was not among the available edges considered by Alg.~\ref{alg:matching1} during $p_j$'s turn, and since $p_j$ has not been assigned to any node in $C$ yet, that means there exists a blocking edge $(p_i,c) \in M'$ with $i < j$ that has been added to $M'$ before $p_j$'s round. Define $f(e) = (p_i,c)$ then. 
    \item[(b)] The complementary case is that $e$ was not selected in $M'$ during $p_j$'s turn but it was still available to pick (that is, $e$ was not blocked). In this situation, Alg.~\ref{alg:matching1} picks the edge with highest weight locally and since $(p_j,c) \not\in M'$ there must be another edge $(p_j,c') \in M'$  with $c' \not= c$ with a higher weight that has been selected instead. Define $f(e) = (p_j,c')$ in this case.
\end{itemize}
\end{itemize}


By exhaustion of possible cases, every edge of $M$ has an image in $M'$. We now prove that every edge $(p_i,c) \in M'$ has at most two preimages under the function $f$. If $(p_i,c) \in M$, the edge has only itself as preimage as this implies that there exist no edges in $M'$ such that $(p',c) \in M'$ with $p' \not= p_i$ nor $(p_i,c') \in M'$ with $c' \not= c$ as $M'$ is a matching of the edges; in this case, $f$ is prevented from applying Cases 2a and 2b and only Case 1 remains.
Now, consider $(p_i,c) \not\in M$. We show that there is only a single edge $e \in M$ such that Case 2a applies so that $f(e) = (p_i,c)$, and the same for Case 2b. For Case 2a to apply, $e$ must be of the form $(p_j,c)$ with $j > i$ and since $M$ is a matching there cannot exist another edge in $M$ containing node $c$. Similarly, for Case 2b to apply, we need to have $(p_i,c') \in M$ and for the same reason there cannot be another edge in $M$ sharing the node $p_i$.

Note that in Case 1, we have trivially $w(e) \leq w(f(e))$; in Case 2a, we have $w(e) \leq \beta \cdot w(f(e))$ by direct application of Assumption~\ref{assumption_beta}; in Case 2b, we have $w(e) \leq w(f(e))$ as the algorithm chooses $f(e)$ as the local maximum of unblocked edges and both $e$ and $f(e)$ are then unblocked. Hence, we can now bound the total weight of the matching $M$ by a sum of weights from edges of $M'$ as follows:

\vspace{-0.5cm}
\begin{equation} \label{eq:main_bound}
  \sum_{e \in M} w(e) \leq \! 
  \underbrace{\sum_{\substack{e \in M, \\ f(e) = e \\~\vspace{0.22cm}}} w(f(e))}_\text{Case 1}  + 
  \underbrace{\sum_{\substack{e=(p_j,c) \in M,\\ f(e) = (p_i,c),\\ i<j \vspace{0.07cm}}} \beta \cdot w(f(e))}_\text{Case 2a} + 
  \underbrace{\sum_{\substack{e=(p_j,c) \in M,\\ f(e) = (p_j,c'),\\ c \not= c'}} w(f(e))}_\text{Case 2b}.  
\end{equation}

As each edge $e' \in M'$ appears either only in the first sum, or at most once in each of the last previous two sums (see discussion above) and $1 + \beta \geq 1$, we get:
$$\sum_{e \in M} w(e) \leq \sum_{e' \in M' \; | \; \exists e \in M, f(e) = e'} (1 + \beta) \cdot w(e').$$

As all weights are greater than zero, we get at last: $$w(M) \leq \sum_{e' \in M'} (1 + \beta) \cdot w(e') = (1 + \beta) \cdot w(M').$$

\end{proof}

As a remark, one subtlety in the above proof is that a blocked edge $e$ may be \emph{lighter} (\ie{} of lower weight) than the blocking edge $e'$ it is mapped with. 
If $w(e) \leq w(e')$, one may wonder why $e$ is part of the optimal solution instead of $e'$. The intuition is that adding $e'$ is not necessarily a good \emph{global choice}. 
Indeed, if $e$ is in $M$ but not in $M'$, one can prove that there is a second edge $e''\in M$ which was not selected by $M'$ and with a cumulative weight $w(e)+w(e'')$ greater than $w(e')$ such that a) either the blocking edge $e'$ is also directly blocking $e''$, or b) $e''$ is at distance at most $2$ from $e$. If no such edge $e''$ exists in $M$, then one can freely swap $e'$ and $e$ in $M$ and improve the optimal matching.

Following Proposition~\ref{theorem_beta}, if $\sigma_P$ implies that for each $c \in C$, the local neighborhood of $c$ is \textit{totally ordered} by considering the edges in the order provided by $\sigma_P$, i.e., $\sigma_P$ is so that $\beta \leq 1$, then Alg.~\ref{alg:matching1} provides a better approximation than the usual greedy algorithm.

We note that this first algorithm may already reduce significantly the number of computed weights, as blocked edges as well as vertices left with a single available edge do not trigger weight computation during its execution. However, in the worst case, the algorithm does end up computing almost all weights in $G$. For instance if $n = \np = \nc$ and $G$ is the complete bipartite graph, $n + (n-1) + \dots + 2 = \frac{n(n+1)}{2}-1 = \Omega(n^2)$ weights are eventually calculated. Even worse, if one strips from the complete bipartite graph all edges that will eventually get blocked by the greedy choices, then a single weight calculation is actually saved.

\begin{remark}
There exist instances in which Alg.~\ref{assumption_beta} computes $\Omega(n^2)$ weights in the worst case.
\end{remark}

The example input given in Figure~\ref{fig:counter_example_beta} illustrates that there exist instances where Alg.~\ref{alg:matching1} reaches its proven approximation bound.

\begin{proposition} \label{counter_example_beta}
Under $\beta$-strong $P$-order, there exist instances where Alg.~\ref{alg:matching1} has an approximation ratio of at least $1 + \beta$.
\end{proposition}

\begin{proof}
Let us consider the example input given by Figure~\ref{fig:counter_example_beta} with $P = \{p_1, p_2\}$ and $C = \{c_1, c_2\}$. Assumption~\ref{assumption_beta} holds on this input as we have $w(p_2,c_1) \leq \beta \cdot w(p_1,c_1)$ and this is the only pair of edges where it can apply.
Alg.~\ref{alg:matching1} selects as matching the pair $(p_1,c_1)$ as it is the local maximum of $p_1$ with weight $1$ (tying with $(p_1,c_2)$ and tie resolution favors $c_1$), to compare with the optimal matching which selects the two other edges with total weight $1+\beta$.

\end{proof}

We introduce the following ``$\gamma-$strong $C$-order'' as the symmetric assumption analogous to Assumption~\ref{assumption_beta} but reversing the sets $P$ and $C$.

\begin{assumption} \label{assumption_gamma}
($\gamma-$strong $C$-order)
Assume $\gamma \geq 0$ and that the set $C$ is provided in the order $\sigma_C =$ $c_1, c_2, \dots, c_{\nc}$, so that for any $c_i,c_j \in C$ with $1 \leq i < j \leq \nc$ and $p \in P$ such that $(p,c_i) \in E$ and $(p,c_j) \in E$, we have $ w(p, c_j) \leq \gamma \cdot  w(p, c_i)$. 
\end{assumption}

\begin{remark} \label{remark:gamma}
If one runs Alg.~\ref{alg:matching1} with input $G= (C\cup P, E)$, i.e.,  inverting the set $P$ and the set $C$ in its input, Assumption~\ref{assumption_gamma} entails that the output is a $(1+\gamma)$-approximation over $G = (P \cup C, E)$ by following Proposition~\ref{theorem_beta} with $\beta = \gamma$. Using the same inputs, a lower bound for the approximation ratio of $1+\gamma$ is also obtained by applying Proposition~\ref{counter_example_beta}.
\end{remark}

Observe that Alg.~\ref{alg:matching1} is not symmetric in $P$ and $C$ and the output that is produced in Remark~\ref{remark:gamma} is naturally different than the one using the original inputs.
Also, when both strong ordering assumptions hold, we can bound Case 2b in Equation~\ref{eq:main_bound} by $\gamma \cdot w(f(e))$; for $\gamma \geq 1$, this worsens the bound but for $\gamma < 1$, we obtain a ratio of $\max \{1, \beta+\gamma \}$ and the max is due to the case $\beta+\gamma$ being smaller than $1$ (i.e., the bound coming from Case 1 is worse then).

\begin{remark} \label{remark:min_matching} 
Under both assumptions $\beta-$strong $P$-order and  $\gamma-$strong $C$-order, for any $\beta > 0$ and $\gamma > 0$, Alg.~\ref{assumption_beta} reaches a $\min\{ 1+\beta, \max\{1, \beta+\gamma\} \}$ approximation.
\end{remark}

Let us note that, according to the example used in Proposition~\ref{counter_example_beta}, the above ratio is reached by Alg.~\ref{alg:matching1} if for example $\gamma > 1$.
By the precedent remark, whenever $\beta+\gamma < 1$, Alg.~\ref{alg:matching1} produces an optimal matching.
One may also deduce from the previous remark that running twice Alg.~\ref{assumption_beta}, first with $G_1 = (P\cup C, E)$ and then with $G_2 = (C\cup P, E)$, and keeping the matching whose weight is maximum produces a $\min\{1+\beta,1+\gamma,\max \{1, \beta+\gamma \}\}$-approximation.

\subsubsection{The \naivealgo{} Algorithm}

\medskip
\begin{figure*}[t]
      \centering
\begin{minipage}{0.32\linewidth}
    \centering
    \begin{tikzpicture}[xscale=2,yscale=0.4]
        \node (p1) at (0,0) {$p_1$};
        \node (c1) at (2,1) {$c_1$};
        \node (c2) at (2,-1) {$c_2$};
        \node (p2) at (0,-2) {$p_2$};
        \draw[-] (p1) -- node[above] {$1$} (c1);
        \draw[-] (p1) -- node[below, pos=0.75]  {$1$} (c2);
        \draw[-] (p2) -- node[below, pos=0.25]  {$\beta$} (c1);
    \end{tikzpicture}
    \caption{Example for Alg.~\ref{alg:matching1}.}
    \label{fig:counter_example_beta}
\end{minipage}%
\hfill%
\begin{minipage}{0.32\linewidth}
    \centering
    \begin{tikzpicture}[xscale=2,yscale=0.4]
        \node (p1) at (0,0) {$p_1$};
        \node (c1) at (2,0) {$c_1$};
        \node (c2) at (2,-2) {$c_2$};
        \node (p2) at (0,-2) {$p_2$};
        \draw[-] (p1) -- node[above,sloped] {$1$} (c1);
        \path[-] (p1) edge[bend left] node[below right]  {$\gamma$} (c2);
        \path[-] (p2) edge[bend left] node[below]  {$\beta$} (c1);
        \draw[-,transparent] (p2) -- node[below,sloped]  {$\varepsilon$} (c2);
    \end{tikzpicture}
    \caption{Example for Alg.~\ref{alg:matching2}.}
    \label{fig:counter_example_gamma}
\end{minipage}
\hfill%
\begin{minipage}{0.32\linewidth}
    \centering
    \begin{tikzpicture}[xscale=2,yscale=0.4]
        \node (p1) at (0,0) {$p_1$};
        \node (c0) at (2,-0.55) {$\vdots$};
        \node (c1) at (2,1) {$c_1$};
        \node (c2) at (2,-2) {$c_{\ell+2}$};
        \node (p2) at (0,-2) {$p_2$};
        \draw[-] (p1) -- node[above, near start] {$1$} (c1);
        \draw[-] (p1) -- node[above, very near end]  {$0.5$} (c0);
        \draw[-] (p1) -- node[above, very near end]  {$\gamma_\ell$} (c2);
        \draw[-] (p2) -- node[below, near start]  {$\beta$} (c1);
    \end{tikzpicture}
    \caption{Example for Alg.~\ref{alg:matching3}.}
    \label{fig:counter_example_betal}
\end{minipage}
\end{figure*}

Previously, the introduced strong ordering assumptions allow to make greedy choices during the processing of nodes by the matching algorithm, however, they do not always guarantee that one can omit the computation of the weight of a single edge of the input graph whenever the assumptions are used separately. 
For instance for Assumption~\ref{assumption_beta}, consider an arbitrarily large graph where each $p_i$, for $1 \leq i \leq \np$, is only connected to two nodes $c_{2i}$ and $c_{2i+1}$ and nothing else, hence omitting the computation of a single weight of the graph may lead to an unbounded approximation as the ordering assumption does not provide bounds on the omitted weight. 
Other problematic instances include star-graphs around a single node $p_1$ (as in Proposition~\ref{claim:oracle_needed}) where Assumption~\ref{assumption_beta} does not provide any constraints on the weights.
Observe that the same argument applies in a symmetric manner with Assumption~\ref{assumption_gamma}.

\begin{remark} \label{algo:infinity}
    By the above arguments and Proposition~\ref{claim:oracle_needed}, even under $\beta$-strong $P$-order (resp.~$\gamma$-strong $C$-order), there exist instances where for any algorithm $\mathcal{A}$ such that $\mathcal{A}$ omits at least the computation of one weight of the input, $\mathcal{A}$ does not produce a bounded approximation.
\end{remark}

However, if both previously introduced assumptions hold in the oracle's orders $\sigma_P$ and $\sigma_C$ simultaneously, then one can actually design an algorithm (Alg.~\ref{alg:matching2}) computing no weights at all but achieving a bounded approximation of the optimal matching. The algorithm simply picks at each step the edge made of the first available and selectable (i.e. having still unmatched neighbors) node $p$ in $\sigma_P$ order paired with the first available node in $p$'s neighborhood, according to $\sigma_C$ order.  Following a similar proof as in Proposition~\ref{theorem_beta}, one derives (Proposition~\ref{claim:beta_gamma}) that if both strong ordering assumptions hold, then Alg.~\ref{alg:matching2} produces a $\max \{ 1, \beta + \gamma \}$-approximate matching without calculating any weights of the input. Using the example of Figure~\ref{fig:counter_example_gamma}, we also show that any matching algorithm that calculates no weights (and in particular Alg.~\ref{alg:matching2}) cannot beat this approximation bound. 

\begin{proposition} \label{claim:beta_gamma}
Under both $\beta$-strong $P$-order and $\gamma$-strong $C$-order, Alg.~\ref{alg:matching2} outputs a $\max \{ 1, \beta + \gamma \}$-approximate matching without calculating any weights.
\end{proposition}

\begin{proof}
The proof follows the same structure as the one of Proposition~\ref{theorem_beta}.
The difference is only that the \naivealgo{} algorithm assigns the first unblocked edge (in $C$'s provided order $\sigma_C$) to $p_j$ whereas the \greedylocal{} algorithm chooses the local maximum of the unblocked edges.
Hence, we can define similarly $f$ and we have again that any edge of $M'$ can only be the image by $f$ of at most two different preimages.
By using the same arguments, the same inequalities on weights hold for Cases 1 and 2a.
Observe now that in Case 2b with $e = (p_j,c_x) \in M$ and $f(e) = (p_j,c_y) \in M'$, we have $x > y$ as $c_y$ is chosen by $M'$ as the first available edge, hence, we have that $w(e) \leq \gamma \cdot w(f(e))$ following Assumption~\ref{assumption_gamma}. 

Summing the edges of $M$ with the three possible subcases, we get:
$$w(M) = \sum_{e \in M} w(e) \leq \! \underbrace{\sum_{\substack{e \in M, \\ f(e) = e \\~\vspace{0.22cm}}} w(f(e))}_\text{Case 1} + \underbrace{\sum_{\substack{e=(p_j,c) \in M,\\ f(e) = (p_i,c),\\ i<j ~\vspace{0.07cm}}} \beta \cdot w(f(e))}_\text{Case 2a} + 
\underbrace{\sum_{\substack{e=(p_j,c) \in M,\\ f(e) = (p_j,c'),\\ c \not= c'}} \gamma  \cdot w(f(e))}_\text{Case 2b}.$$
With analogous concluding arguments to the ones in the proof of Proposition~\ref{theorem_beta}, we get that each edge of $M'$ can either be also present in $M$ and has then a unique image by $f$, or appear at most once in each Cases 2a and 2b, entailing:
$$w(M) \leq \sum_{\substack{e' \in M' \\ \exists e \in M, f(e) = e'}} \max \{ 1, \beta + \gamma \} \cdot w(e') \leq \max \{ 1, \beta + \gamma \} \cdot w(M').$$
\end{proof}

As with Remark~\ref{remark:min_matching}, it is interesting to note that the previous proof also shows the optimality of the algorithm for some strong heuristic orders on the input nodes. 

\begin{remark} \label{remark:naive_optimal}
     Without computing any weights, the \naivealgo{} matching algorithm is optimal under $\beta$-strong $P$-order and $\gamma$-strong $C$-order whenever $\beta + \gamma \leq 1$.
\end{remark}

We can also show the above remark by a constructive, direct and more intuitive proof. First consider the following lemma:

\begin{lemma} \label{lemma:opt_partition}
Suppose a graph $G = (V,E,w)$ with $E = E_1 \cup E_2$ and $E_1 \cap E_2 = \emptyset$. Denote $M^1_{opt}$ the optimal matching over $G_1 = (V, E_1)$ and $M^2_{opt}$ the optimal matching over $G_2 = (V, E_2)$, and $M_{opt}$ the one over the full graph $G = (V,E)$.
Then $w(M_{opt}) \leq w(M^1_{opt}) + w(M^2_{opt})$.
\end{lemma}

\begin{proof}
Split the edges of $M_{opt}$ into two subsets $M_1$ and $M_2$ according to the edge partition of $G$, i.e. $M_1 = M_{opt} \cap E_1$ and $M_2 = M_{opt} \cap E_2$.
We have $w(M_{opt}) = w(M_1) + w(M_2)$ and since $M_i$ is a valid matching over $G_i$, we have $w(M_i) \leq w(M^i_{opt})$ thus $w(M_{opt}) \leq w(M^1_{opt}) + w(M^2_{opt})$.

\end{proof}

Now observe that, whenever $\beta+\gamma < 1$, the edge $e$ that is greedily selected by Alg.~\ref{alg:matching2} is optimal in its ``neighborhood'' $N_e$ (that is all possible paths of $3$ edges with $e$ in central position). Hence, applying Lemma~\ref{lemma:opt_partition} with $E_1=N_e$ and $E_2=E \setminus N_e$, one can show by induction that the \naivealgo{} matching algorithm is optimal in this case.

Also, one may note that adding both strong order assumptions with $\beta < 1$ and $\gamma < 1$ gives a strict total order on each of the neighborhoods, for all nodes in $P$ and in $C$. However, it is noteworthy to mention that even in this situation with strong starting assumptions, the edge ordering is still \textit{partial} and ``weaker'' than a total edge ordering (which is required by the classic $2$-approximate greedy algorithm that scans all the edges in decreasing weight order), as stated in the following remark.

\begin{remark} \label{remark:comparable_edges}
    Even under $\beta$-strong $P$-order and $\gamma$-strong $C$-order and if both $\beta < 1$ and $\gamma < 1$, there exist pairs of edges for some input graphs that are incomparable before requesting the weight of the respective edges.
\end{remark}

In particular, one can consider any pair of edges not sharing any endpoint and such that each edge of the pair is the first of its neighborhood for both its endpoints, then for both considered edges, their respective weight is entirely unbounded by the ordering assumptions. Thus, under both strong order assumptions (but that do not enforce a total edge ordering), the aforedefined naive ``no weight calculations'' algorithm outputs a matching with a better approximation guarantee than the usual greedy algorithm. 

At last, observe that when $\gamma<1$, all edges $(p,c_i)$ are ordered by decreasing weights, implying that the first available edge in the provided $C$-order is also the local maximum according to $p$. Hence Alg.~\ref{alg:matching2} is equivalent to Alg.~\ref{alg:matching1} in this situation. Following the result of Proposition~\ref{theorem_beta}, we may observe the following.

\begin{remark} \label{remark:degenerate_case} 
If both $\beta$-strong $P$-order and $\gamma$-strong $C$-order assumptions hold and $0 \leq \gamma <1$, Alg.~\ref{alg:matching2} produces the same matching as Alg.~\ref{alg:matching1} without calculating any weights.
\end{remark}

At last, we note that there cannot exist a better algorithm than Alg.~\ref{alg:matching2} in terms of approximation ratio when no weights are accessed.

\begin{proposition} \label{prop:counter_example_alg2}
Under both $\beta$-strong $P$-order and $\gamma$-strong $C$-order, any matching discovery algorithm that calculates 0 weights cannot be better than ($\beta+\gamma$)-approximate.
\end{proposition}

\begin{proof}
Let us consider the 4-nodes instance given by Figure~\ref{fig:counter_example_gamma}. 
Given the provided ordering of vertices in $P$ and $C$, we have that both Assumptions~\ref{assumption_beta} and \ref{assumption_gamma} hold on the instance.
Obviously, any algorithm cannot provide better than a $1$-approximation so let us assume $\beta + \gamma \geq 1$.
Note first that the Naive-Local matching on this instance produces $\{(p_1,c_1)\}$ with weight $1$ whereas the optimal picks the two other edges with weight $\beta + \gamma$.
Now, consider a matching algorithm $\mathcal{A}$ that picks $(p_1,c_2)$ and $(p_2,c_1)$. In that case, change the instance so that $w(p_1,c_1) = \alpha$ with $\alpha$ arbitrarily large and all other weights set to $1$ to simplify (note that both our underlying assumptions still hold in this situation as well). $\mathcal{A}$ is then arbitrarily far from the optimal matching that selects $(p_1,c_1)$.

\end{proof}

\subsubsection{The \boundedgreedy{} Algorithm}

\medskip
Our first results show that the first set of assumptions that was considered may be unsatisfactory for two reasons: either one of the assumptions holds and all weights may end up being computed or both assumptions hold at the same time and absolutely no weight calculations are required to reach a bounded approximation ratio. This may indicate that the assumptions could be too strong in some sense. We design here weaker assumptions that only require the condition on one set to hold (e.g., Assumption~\ref{assumption_beta}) and a weaker and \textit{more local} form of the other assumption: the bound holds between node $p \in P$ and $c,c' \in C$ if there exist at least $\ell$ other neighbors of $p$ between $c$ and $c'$ when taken in $\sigma_C$ order. That is, we do not control the weight of successive edges in a given node's neighborhood but if there are $\ell$ other edges $(p,c_j)$ between two edges $(p,c_0)$ and $(p,c_{\ell+1})$, then the latter one must have a bounded weight in comparison to $(p,c_0)$. The following assumption allows us to design a matching algorithm (Alg.~\ref{alg:matching3}) requiring only at most $\ell + 1$ weight computations for each node in $P$. 

\begin{assumption} \label{assumption_gamma_l}
($\gamma_\ell$-$\ell$-weak $C$-order)
Assume $\ell \geq 0$, $\gamma_\ell \geq 0$ and $\sigma_C = c_1, c_2, \dots, c_{\nc}$, so that for any $c_i,c_j \in C$ with $1 \leq i < j \leq \nc$
and $p \in P$ such that $(p,c_i),(p,c_j) \in E$ and $|\{(p,c_x) \in E \; | \; i < x < j \}| \geq \ell$, we have $w(p, c_j) \leq \gamma_\ell \cdot  w(p, c_i)$.
\end{assumption}

In the above assumption, smaller values for $\ell$  make the assumption stronger, with $\ell = 0$ being equivalent to $\gamma$-strong $C$-order (i.e. Assumption~\ref{assumption_gamma} with $\gamma = \gamma_0$) and $\ell = \Delta(G_C)-1 = \max_{c \in C} \delta(c)-1$ with $\delta(c)$ the degree of node $c$ (number of edges in $E$ with $c$ as endpoint) being always true for any input graph $G = (P \cup C, E)$.
For fixed processing orders on the nodes, the value of $\gamma_\ell$ decreases as $\ell$ increases and reaches its ``(potentially non-zero) minimum'' at $\Delta(G_C)-2$ (after which $\gamma_\ell = 0$ as the bound requirement does not apply to any pair of edges).
Introducing a weak order allows to add weaker constraints on the edge weights than the ones implied by strong orders. However, obviously weak orders for $\ell \geq 1$ do not help when no weights are ever computed as they do not provide bounds for some edges sharing endpoints (hence any choice between the two may entail an arbitrarily large error). 
For instance, using the example of Figure~\ref{fig:counter_example_gamma} under $\gamma_1$-$1$-weak $C$-order, $w(p_1,c_2)$ can take arbitrarily large values and  Alg.~\ref{alg:matching2} selects $(p_1,c_1)$ on this instance.

\begin{remark} \label{remark:alg2_infinity}
Under both $\beta$-strong $P$-order and $\gamma_\ell$-$\ell$-weak $C$-order (resp. $\gamma$-strong $C$-order and $\beta_\ell$-$\ell$-weak $P$-order), there are instances where Alg.~\ref{alg:matching2} has infinite approximation ratio.
\end{remark}


Let us show how we design an efficient discovery algorithm (Alg.~\ref{alg:matching3}) by exploiting the assumption of a strong order $\sigma_P$ over one partition and a weak order $\sigma_C$ on the other one. The algorithm we introduce is similar in flavor to the first defined algorithm, but this time, instead of taking the edge with maximum weight over the full neighborhood of $p_i$, only the $\ell+1$ first available edges according to $\sigma_C$ are considered. 

\begin{proposition} \label{prop_gamma_l} 
Under both $\beta$-strong $P$-order and $\gamma_\ell$-$\ell$-weak $C$-order, Alg.~\ref{alg:matching3} has approximation ratio at most $\max\{1+\beta, \beta+\gamma_\ell \}$.
\end{proposition}

\begin{proof}
The proof follows the same structure as the one for Proposition~\ref{theorem_beta}.
Define $M$ as an optimal matching, $M'$ as the matching produced by Alg.~\ref{alg:matching3} on $G$, and define similarly as previously $f$ as a mapping of $M$'s edges into $M'$ with identical Cases 1 and 2a.
For Case 2b, that is when we consider an edge $e = (p,c) \in M$ such that $e \not\in M'$ while considering that $(p,c)$ is unblocked during $p$'s assignment round, we define $f(e)$ as the edge with the maximum weight among the $\ell+1$ first unblocked edges (in the same way as Alg.~\ref{alg:matching3} picks the edge during $p$'s round).
Since for each $p$, we assign as before an edge of its neighborhood by $f$, our previous arguments hold regarding the number of preimages by $f$.
Now, consider the bound on the weight of edges in $M$. We know that $w(e) \leq w(f(e))$ in Case 1 (trivial) and $w(e) \leq \beta \cdot w(f(e))$ in Case 2a following Assumption~\ref{assumption_beta}. 

In Case 2b, let us consider two possible subcases. 

(1) If there are at most $\ell+1$ unblocked edges during $p$'s round, then since $e$ is unblocked, it is among those edges. Hence, by the property that $w(f(e))$ is the maximum of the weights of the unblocked edges, we get $w(e) \leq w(f(e))$.

(2) Suppose there are strictly more than $\ell+1$ unblocked edges. 
Since if $e$ were among the first $\ell+1$ ones we would also have $w(e) \leq w(f(e))$, let's assume $e = (p,c_j)$ is not among these edges. By Assumption~\ref{assumption_gamma_l}, recall that one cannot bound the weights of the edges between $p$ and its neighbors $c_i$ such that $(p,c_i)$ is among the $\ell$ distinct edges incident to $p$ directly preceding $(p,c_j)$ in $\sigma_C$ order; note $T(c_j)$ the set of these edges. 
If Alg.~\ref{alg:matching3} selects an edge $f(e) = (p,c_x)$ outside $T(c_j)$, we can apply the aforementioned assumption and get $w(e) \leq \gamma_\ell \cdot  w(f(e))$; recall here that $f(e)$ is among the first $\ell+1$ available edges of $p$'s neighborhood hence in particular, it cannot be placed \textit{after} $c_j$ in $\sigma_C$ order.
Thus, let's suppose herafter that Alg.~\ref{alg:matching3} selects an edge $(p,c_x) \in T(c_j)$.
Observe that among the $\ell$ edges of $T(c_j)$, some of them might be blocked and others unblocked. In any case, among the first $\ell+1$ edges that are considered by the algorithm, there exists at least one unblocked edge $(p,c') \not\in T(c_j)$ because $|T(c_j)| = \ell$ and we assumed at least $\ell + 2$ unblocked edges in $p$'s neighborhood.
Finally, we have $w(p,c') \leq w(p,c_x)$ because the algorithm picked the edge with the best weight, and thus $w(p,c_j) \leq \gamma_\ell \cdot w(p,c')$ by application of Assumption~\ref{assumption_gamma_l}, which gives us  $w(e) \leq \gamma_\ell \cdot w(f(e))$ in this case as well.

Putting together the two subcases for Case 2b, we have $w(e) \leq \max\{1,\gamma_\ell\} \cdot w(f(e))$.
By reusing analogous arguments as in the proofs of Proposition~\ref{theorem_beta} and Proposition~\ref{claim:beta_gamma}, we get $w(M) \leq \max \{1, \beta + \max\{1,\gamma_\ell\} \} \cdot w(M')$. Since $\beta \geq 0$ and $\max\{1,\gamma_\ell\} \geq 1$, we get $w(M) \leq \max\{1+\beta, \beta+\gamma_\ell \} \cdot w(M')$.

\end{proof}

\begin{proposition} \label{prop:alg2_weights} 
Alg.~\ref{alg:matching3} calculates at most $(\ell + 1) \cdot n$ weights.
\end{proposition}

\begin{proof}
Note first that at each of the $s$ iterations of the algorithm, at most $\ell + 1$ weights are calculated. 
Also, if at least one weight is calculated at a given iteration, then an edge is added to the constructed matching. 
Since at most $n$ edges may ever be added to the matching, there are only $n$ iterations where at least one weight is calculated.

\end{proof}

\begin{proposition} \label{prop:alg3_strong}
Under both $\beta$-strong $P$-order and $\gamma$-strong $C$-order, Alg.~\ref{alg:matching3} is $\max\{1, \beta+\gamma \}-$approximate.
\end{proposition}

\begin{proof}
If $\gamma \geq 1$, the proposition follows directly from Proposition~\ref{prop_gamma_l} for the case $\ell = 0$.
For $\gamma < 1$, following the same argument as in Remark~\ref{remark:degenerate_case}, Alg.~\ref{alg:matching3} degenerates and produces the same solution as the \naivealgo{} algorithm.

\end{proof}

The example of Figure~\ref{fig:counter_example_betal} can be used to show that under Assumptions~\ref{assumption_beta} and~\ref{assumption_gamma_l}, Alg.~\ref{alg:matching3} reaches its proven approximation bound.

\begin{proposition} \label{prop:alg3_counter_example}
Under both $\beta$-strong $P$-order and $\gamma_\ell$-$\ell$-weak $C$-order, Alg.~\ref{alg:matching3} has approximation ratio at least $\max\{1+\beta, \beta+\gamma_\ell \}$.
\end{proposition}

\begin{proof}
Consider first the example of Figure~\ref{fig:counter_example_beta}. 
Following the same arguments as in the proof of Proposition~\ref{counter_example_beta}, we get that Alg.~\ref{alg:matching3} (which is equivalent to Alg.~\ref{alg:matching1} on that example) produces a $(1+\beta)$-approximate matching.
Suppose $\beta + \gamma_\ell > 1 + \beta$, that is $\gamma_\ell > 1$, and let us use the example of Figure~\ref{fig:counter_example_betal} where $p_1$ has $\ell+2$ neighbors with $w(p_1,c_j) = 0.5$ for $2 \leq j \leq \ell+1$.
In this example, Assumption~\ref{assumption_beta} only applies to $(p_1,c_1)$ versus $(p_2,c_1)$ and Assumption~\ref{assumption_gamma_l} to $(p_1,c_1)$ versus $(p_1,c_{\ell+2})$.
In the example, the algorithm picks $(p_1,c_1)$ for a weight of $1$ whereas the optimal matching picks $(p_1, c_{\ell+2})$ and $(p_2,c_1)$ for a weight of $\beta + \gamma_\ell > 1$.

\end{proof}

We use Assumption~\ref{assumption_beta_l} to obtain symmetric results (Proposition~\ref{prop_beta_l}).

\begin{assumption} \label{assumption_beta_l}
($\beta_\ell$-$\ell$-weak $P$-order) Assume $\ell \geq 0$, $\beta_\ell \geq 0$ and $\sigma_P = p_1, p_2, \dots, p_{\np}$, so that for any $p_i,p_j \in P$ with $1 \leq i < j \leq \np$
and $c \in C$ such that $(p_i,c)\in E$ and $(p_j,c) \in E$ and such that $|\{(p_x,c) \in E \; | \; i < x < j \}| \geq \ell$, we have $w(p_j, c) \leq \beta_\ell \cdot  w(p_i, c)$.
\end{assumption}

\begin{proposition} \label{prop_beta_l}
Under both $\gamma$-strong $C$-order and $\beta_\ell$-$\ell$-weak $P$-order, Alg.~\ref{alg:matching3} is $\max \{1 + \gamma, \beta_\ell + \gamma \}$-approximate on input $G = (C \cup P, E)$.
\end{proposition}

Inverting $P$ and $C$ in Proposition~\ref{prop:alg3_counter_example}, one can show that the bound in the previous proposition is reached by Alg.~\ref{alg:matching3} on some instances.

\paragraph{Limitation of greedy-choice algorithms} 

We explain briefly here why lifting Assumption~\ref{assumption_beta} controlling the order in which vertices of $P$ are processed and replacing it by a bounded variant tolerating edges that are out-of-order such as Assumption~\ref{assumption_beta_l} leads to impossibility to approximate the optimal matching by a \textit{greedy-choice} algorithm (picking each round the available edge with maximum observed weight). 
As a counter-example, one can consider a path as an instance and can derive that any algorithm inspecting only a bounded number of edges before adding irreversibly the observed edge with greatest weight to the matching (hence, allowing to pick some edges whose neighborhood is not completely explored), may fail to provide a bounded-approximation.
This is due to the fact that the algorithm has no control on the weight of the edges connected to some of the inspected edges on the input path. 
We also note that on a path, both weak assumptions with $\ell \geq 1$ do not apply to any pair of edges and thus all weights are unrestrained in this case.
We note that the formulation of the claim as it appeared in~\cite{duvignau2023greediness} does not make explicit the notion of \textit{greediness} that is being used, which we clarify in the below proposition. 

\begin{proposition} \label{claim:bounded_search} 
Fix $\ell \geq 1$. Suppose $\beta_\ell$-$\ell$-weak $P$-order and $\gamma_\ell$-$\ell$-weak $C$-order hold. Consider now a greedy-choice matching algorithm $\mathcal{A}$ (i.e., that greedily adds the examined edge with highest weight) that always decides to add an edge after querying at most $k_\ell$ edges for some $k_\ell \geq 1$. Then $\mathcal{A}$ does not provide a bounded approximation ratio.
\end{proposition}

\begin{proof}

The idea behind the proof is that in general, it is possible to force having the edge with highest weight neighboring a non-queried edge by the time the algorithm has to greedily add an edge.
To rule out the constraints on weights stemming from the weak orders when $\ell \geq 1$, we assume a graph of degree at most $2$ where the weak ordering assumptions do not apply. 
Without loss of generality, let us consider as a counter example a path made of at least $2^{k_\ell}$ edges.
The weights of the edges forming the path needs to be assigned ``online'' depending in which order $\mathcal{A}$ examines the edges' weight: upon examining the $i$-th edge with $1 \leq i \leq k_\ell$, if $\mathcal{A}$ inspects an edge 
within 
the current longest sub-path made only of undiscovered edges (denoted $P_i$) then we set $w(e) = i$, otherwise $w(e) = 0.1$.
Under these circumstances, let us show 
that the edge $e$ with highest weight after $i \leq k_\ell$ steps is always neighboring an undiscovered edge.
We can show by induction that the highest weight after $i$ edges have been inspected is always adjacent to a sub-path of undiscovered edges of length at least $2^{k_\ell - i}$. This is because, at each step, either the highest weight does not change (an edge outside $P_i$ was discovered) or it changes and it splits $P_i$ in two parts $S_1$ and $S_2$ with $\max \{|S_1|, |S_2|\} \geq P_i/2$. After $k_\ell$ inspections, the highest weight is thus adjacent to a sub-path of undiscovered edges of size at least $2^0 = 1$.
To conclude, we set an arbitrarily large weight to the adjacent undiscovered edge.


\end{proof}

By the previous claim, it is fruitless in general to try to design a bounded-approximation greedy-choice algorithm that discovers, at each of its iteration, a bounded number of edges. 
However, as we show in the next section, it is possible to design a bounded-approximation algorithm assuming weak orders on both input sets and that uses only on average a bounded number of discovery queries per edge in the output matching.



\subsubsection{Double-Greedy Algorithm} \label{subsec:double_greedy_algo}

\begin{figure*}[t]
    \centering
    \begin{tikzpicture}[yscale=0.45]
        \node (p1) at (0,0) {$p_1$};
        \node (c1) at (3,1) {$c_1$};
        \node (c2) at (3,-1) {$c_2$};
        \node (p2) at (0,-2) {$p_2$};
        \node (c3) at (3,-3) {$c_3$};
        \node (p3) at (0,-4) {$p_3$};
        \node (c4) at (3,-5) {$c_4$};
        \draw[-, ] (p1) -- node[midway, above] {$7$} (c1);
        \draw[-, ] (p1) -- node[near end, above]  {$8$} (c2);
        \draw[-,] (p1) -- node[midway, above]  {$9$} (c3);

         \draw[-, ] (p2) -- node[above, very near start, sloped] {$1$} (c1);
        \draw[-,] (p2) -- node[above, very near start, sloped] {$8$} (c3);
         \draw[-, ] (p2) -- node[below, very near start, sloped] {$3$} (c4);

         \draw[-, ] (p3) -- node[above, very near end, sloped] {$4$} (c2);
         \draw[-,] (p3) -- node[above, midway, sloped] {$7$} (c4);

         \node[above,font=\bfseries] at (current bounding box.north) {(a) Ordered Input};
    \end{tikzpicture}
    \hspace{0.5cm}%
    \begin{tikzpicture}[yscale=0.5,xscale=1.5]
        \node (p1) at (0,0) {$p_1$};
        \node (c2) at (1,0) {$c_2$};
        \node (p3) at (2,0) {$p_3$};
        \node (c4) at (3,0) {$c_4$};
        \node (p2) at (4,0) {$p_2$};
        \node (c3) at (5,0) {$c_3$};
        
        \node[gray] (c1) at (1,-1) {$c_1$};

        \draw[->, thick] (p1) -- node[midway, above]  {$8$} (c2);
        \draw[<-, thick] (p3) -- node[midway, above] {$4$} (c2);
        \draw[->, thick] (p3) -- node[midway, above] {$7$} (c4);
        \draw[<-, thick] (p2) -- node[midway, above] {$3$} (c4);
        \draw[->, thick] (p2) -- node[midway, above] {$8$} (c3);

        \path[->, dashed, gray] (c3) edge[in=60,out=120] node[midway, above]  {$9$} (p1);
        
        \draw[-, dotted, gray] (p1) -- node[midway, below] {$7$} (c1);
        \draw[-, dotted, gray] (p2) -- node[midway, below] {$1$} (c1);
        
         \node[above,font=\bfseries] at (current bounding box.north) {(b) Greedy Path};

         \node (sp1) at (0,-2.5) {$p_1$};
        \node (sc2) at (1,-2.5) {$c_2$};
        \node (sp3) at (2,-2.5) {$p_3$};
        \node (sc4) at (3,-2.5) {$c_4$};
        \node (sp2) at (4,-2.5) {$p_2$};
        \node (sc3) at (5,-2.5) {$c_3$};

        \draw[-, ultra thick] (sp1) -- node[midway, above]  {$\checkmark$} (sc2);
        \draw[-, dashed, lightgray] (sp3) -- node[midway, above] {$\times$} (sc2);
        \draw[-, ultra thick] (sp3) -- node[midway, above] {$\checkmark$} (sc4);
        \draw[-, dashed, lightgray] (sp2) -- node[midway, above] {$\times$} (sc4);
        \draw[-, ultra thick] (sp2) -- node[midway, above] {$\checkmark$} (sc3);

         \node[below,font=\bfseries] at (current bounding box.south) {(c) Selected Edges};
    \end{tikzpicture}

    \caption{Example of execution for Alg.~\ref{alg:double_greedy} using orders $\sigma_P = p_1, p_2, p_3$, $\sigma_C = c_1, c_2, c_3, c_4$ and $\ell = 1$: (a) Ordered input; (b) Constructed greedy path starting from $p_1$ with forward edges as plain lines, non-selected edges (because not local maximum) as dotted and backward edges as dashed; (c) Selected edges as optimal matching over the greedy path. Alg.~\ref{alg:double_greedy} does not run another greedy path procedure as all nodes in $P$ are then made unavailable and outputs $M=\{(p_1,c_2),(p_3,c_4),(p_2,c_3)\}$ with total weight $w(M) = 23$.}
    \label{fig:double_greedy}
\end{figure*}
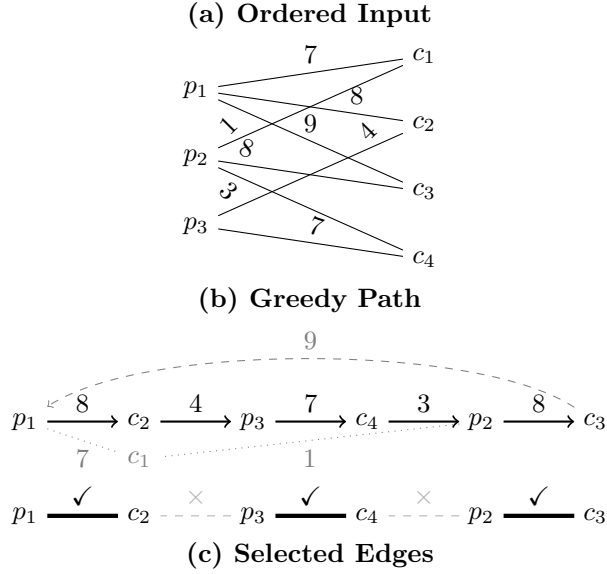

Proposition~\ref{claim:bounded_search} basically indicates that the strategies used so far in order to develop greedy matchings based on heuristic orders fail for the situation where two weak orders are used.
In this situation, a greedy-choice algorithm is not possible in order to reach a bounded approximation and one has to consider an algorithm that explores more of the input before taking even a single decision. 
This is the case for our last greedy algorithm (Alg.~\ref{alg:double_greedy}) that only requires weak node orders to get a bounded approximation ratio.
The algorithm is the following. 
At a given iteration, considering the next node $p_i$ in $\sigma_P$ order that is still available then, it starts first by building a (oriented) \textit{greedy path} $U$ starting initially from node $u = p_i$.
The greedy path is obtained by adding at each step the edge $(u,v)$ with highest weight within $u$'s ``bounded local neighborhood'' $N_\ell(u)$ (i.e., the $\ell+1$ first available edges in the heuristic order related to the node $u$, that is $\sigma_C$ if $u \in P$ and $\sigma_P$ if $u \in C$), and continuing the next attempt to extend the path from node $v$ and so on.
At each step, edges sharing an endpoint with a node that is already part of the constructed path (referred to as ``backward edges'') are discarded from being included in $N_\ell(u)$.
The greedy path ends when the bounded local neighborhood (without backward edges) of the last processed node $u$ is empty.
Then Alg.~\ref{alg:double_greedy} picks an optimal solution \texttt{OptimalPath($U$)} for the path $U$ and adds it to the constructed matching.
The algorithm continues until the exhaustion of available nodes in $P$ from which a greedy path can be initiated.
An example of an execution of the algorithm is illustrated in Figure~\ref{fig:double_greedy}. 
Before proving the correctness of Alg.~\ref{alg:double_greedy}, we demonstrate a more elementary result comparing the weight of the optimal matching $U_{opt}$ over any path $U$ with the weight of the edges in $U \setminus U_{opt}$ that are not selected in the optimal.


\begin{lemma} \label{lemma:path}
    Let $U = e_1, \dots, e_m$ be a sequence of edges defining a path (i.e., for all $2 \leq i \leq m$ with $e_{i} \cap e_{i-1} \not= \emptyset$ and $|e_i \cap \bigcup_{1 \leq j \leq i-1} e_j| = 1$), and $U_{opt}$ the optimal matching over $U$. We have then $w(U_{opt}) \geq \sum_{e\in U \setminus U_{opt}} w(e)$.
\end{lemma}

\begin{proof} 
    Let $U' = U \setminus U_{opt}$. 
    Hereafter for clarity we explicitly state a matching as ``proper'' to refer to a \textit{valid} matching of the edges (i.e., without adjacent edges) and improper for any other set of edges.
    If $U'$ is a proper matching then the result is trivial.
    Let us analyze the different possibilities for $U_{opt}$.
    To ease the notation, we will use the following convention: a subset $M$ of a subpath $e_j, \dots, e_{j+k}$ of $U$ (for $1 \leq j \leq m-k+1$) formed by $k$ consecutive edges is denoted by a word $u_1 \dots u_k$ over the alphabet $\{\times, \checkmark\}$, such that for $1 \leq i \leq k$, $u_i = \checkmark$ if $e_{j+i} \in M$ and $u_i = \times$ otherwise. 
    In the following, we use the word $u=u_1\ldots u_m$ to denote the edges selected by $U_{opt}$ and the word $\overline{u}=\overline{u}_1\ldots \overline{u}_m$ to denote the edges in $U'$.
    First, observe that $U_{opt}$ does not omit three (or more) consecutive edges $e_j$, $e_{j+1}$ and $e_{j+2}$ (a pattern denoted as $\times \times \times$ in $u$) as adding $e_{j+1}$ makes a better matching than $U_{opt}$ in this case (recall weights are strictly positive here).
    $U_{opt}$ may omit two consecutive edges. By the previous argument, both the adjacent edges to the omitted ones must be in $U_{opt}$, i.e., $u$ can contain the pattern $(...) \checkmark \times \times \checkmark (...)$ that we will refer as a \textit{hole} in $U$ and the reversed pattern as an \textit{antihole} in $U'$ (i.e., a subsequence $\times \checkmark \checkmark \times$ in $\overline{u}$ containing two consecutive selected edges). Note the path itself cannot end in a double omission as the last edge could then be freely added to $U_{opt}$.
    Observe that antiholes are the only pattern preventing $U'$ from being a proper matching.
    Now, let us iteratively define a proper matching $U''$ based on $U'$ and with greater weight than $w(U') = \sum_{e \in U'} w(e)$, extending the $w$ notation to improper matchings.
    For this, set $U'_0 = U'$. 
    Let $i \geq 0$. 
    We will remove the last antihole $H_i = e_j, e_{j+1}, e_{j+2}, e_{j+3}$ in $U'_i$ to produce a new set $U'_{i+1}$ such that: (1) $U'_{i+1}$ has one antihole less than $U'_i$, (2) the part of $U'_{i}$ ``before $H_i$'' is identical in $U'_{i+1}$ and (3) $w(U'_{i+1}) \geq w(U'_i)$. 
    Since by hypothesis the edges of $H_i$ have not been ``modified'' in $U'_{i}$, the induction argument is proven by considering how the antihole $H_i$ appears in $U_{opt}$ and what follows it in the optimal matching.

    Formally, let $j$ be the starting position of the antihole $H_i$ and $\overline{u}^{(i)}$ the word representing $U'_i$.
    We construct the next word $\overline{u}^{(i+1)}$ as follows:
    \begin{equation*} 
        \overline{u}^{(i+1)} = \overline{u}^{(i)}_{1}\cdots  \underbrace{\overline{u}^{(i)}_{j} \cdot \overline{u}^{(i)}_{j+1}\cdot u_{j+2} \cdot u_{j+3}}_\text{$\checkmark \times \times \checkmark$ in $U_{opt}$ and $\times \checkmark \checkmark \times$ in $U'_i$}  \cdots u_m.
    \end{equation*} 

    By construction of $\overline{u}^{(i+1)}$, both (1) and (2) hold.
    
    Denote $X_i = \{ e_r \in U_{opt} \;|\; r \geq j+2\}$ and $X'_i = \{ e_r \in U'_i \;|\; r \geq j+2\}$.
    Let us verify now that $w(X_i) \ge w(X'_i)$.
    Suppose the cumulative weight of the edges of $X'_i$ is greater than $w(X_i)$, then since by hypothesis $\overline{u}^{(i)}_{j+2}\cdots \overline{u}^{(i)}_m$ does not contain any antiholes and we have $u_{j+1} = \times$, the set $U_{opt} \setminus X_i \cup X'_i$ is a proper matching with greater weight than $U_{opt}$, hence a contradiction. 
    Thus, $w(U'_{i+1}) = w(U'_i) - w(X'_i) + w(X_i) \geq w(U'_i)$.
    
    Therefore, all three induction hypotheses hold: (1) $\overline{u}^{(i+1)}$ contains one antihole less than $\overline{u}^{(i)}$, (2) $\overline{u}^{(i+1)}_r = \overline{u}^{(i)}_r$ for $1 \leq r \leq j$, and (3) the weight condition $w(U'_{i+1})\ge w(U'_i)$.
    
    Denote $U''$ the set obtained after purging all antiholes from $U'$ by the above procedure, i.e. $U'' = U'_{k}$ where $k$ is the initial number of antiholes in $U'$. By definition, $U''$ is a proper matching of $U$ and we have $w(U'') \geq w(U')$ by induction. This concludes the proof because $U''$ as a proper matching also entails $w(U'') \leq w(U_{opt})$.
    
\end{proof}

\begin{figure*}[t]
      \centering
\begin{minipage}{0.535\linewidth}
\begin{algorithm}[H]
\footnotesize
\SetKwInOut{Input}{Input}
\SetKwInOut{Output}{Output}
\SetKwInOut{Upon}{Event}
   \Input{
A bipartite graph $G = (P \cup C, E)$ with sets $P = p_1, p_2, \dots, p_{\np}$ and $C = c_1,c_2, \ldots, c_{\nc}$
}
\Output{$M$, a matching of $E$\;}
\vspace{0.1cm}
  \SetKwFunction{algo}{next\_edge\_greedy\_path}

    \tcp{Initialization of the procedure}
    $M, i \leftarrow \emptyset, 1$ \; \label{alg1:line1}

    \ForEach{$x \in P \cup C$}{ 
        $\text{available}_x \leftarrow $ True \; \label{alg1:line3}
    }

    \tcp{Loop till all nodes in $P$ are matched}
    \While{$i \leq \np$}
    {
        \eIf{$\text{available}_i$}
        {
            \tcp{initialize the greedy path}
            $U \leftarrow []$ \;
            \tcp{endpoint of the path $U$}
            $u \leftarrow p_i$\;
            \Repeat{
                $v = 0$
            }
            { 
                \tcp{try to extend the path}
                $v \leftarrow$\algo{$u$,$U$}; 
                \If{ $v \not= 0$}{  
                    \tcp{add edge $(u,v)$ to $U$}
                    $U.$\textsf{append}$((u,v))$\;
                    \tcp{continue then from $v$}
                    $u \leftarrow v$\; 
                }
            }   

            \eIf{$|U| = 0$}{
                 $i \leftarrow i + 1$\;
            }{
                $M_p  \leftarrow$ \texttt{OptimalPath}($U$)\; \label{alg4:optimal_path}
                \ForEach{$(p_x,c_y) \in M_p$}{ \label{alg4:selected_edges}
                    $\text{available}_x \leftarrow$ False\;
                    $\text{available}_y \leftarrow$ False\;
                    $M \leftarrow M \cup \{(p_x, c_y)\}$\;
                }
            }

        }{
            $i \leftarrow i + 1$\;
        }
    }
\Return $M$\;
\caption{$\ell$-Double-Greedy Matching} \label{alg:double_greedy} 
\end{algorithm}
\end{minipage}%
\hfill%
\begin{minipage}{0.4625\linewidth}
\begin{algorithm}[H]
\footnotesize
  \setcounter{AlgoLine}{21} 
  \SetKwFunction{algo}{next\_edge\_greedy\_path}
  \SetKwFunction{proc}{proc}
  \SetKwProg{myalg}{Function}{}{}
  \tcp{Uses inputs \& variable $\text{available}_x$ from Alg.~\ref{alg:double_greedy}}
  \myalg{\algo{u,U}}{ \tcp{Find next edge from node $u$ on the greedy path $U$}
   \eIf{$u \in P$}
    {
     $k \leftarrow \nc$\; 
    }{
     $k \leftarrow \np$\; 
    }

    \tcp{Consider all possible ``forward'' edges}
    $N \leftarrow \{ 1 \leq x \leq k  \;|\; \{u,x\} \in E \land  \text{available}_x \land \not\exists y, \{x,y\} \in U \}$ \;  \label{alg4:greedy_condition}

    \eIf{$N \not= \emptyset$}
    {
        \eIf{$|N| > 1$}{
            \tcp{Keep the $\ell+1$ first values}
            $N \leftarrow N[:\ell+1]$ \;
            \ForEach{$j \in N$}{
                \tcp{Get $w(u,x)$ when $u\in P$ otherwise $w(x,u)$}
                $b_x \leftarrow$ \textsf{weight}$(u, x)$\;
            }
            \tcp{As in Alg.~\ref{alg:matching1}, line~\ref{line:argmax}}
            $j \leftarrow \argmax_{x \in N} \; b_{x}$\;
        }
        {
            $j \leftarrow N[1]$; \tcp{1st value} 
        }
    }
    {
     \tcp{$0$ stands for ``end of path''}
     $j \leftarrow 0$\;
    }
    \tcp{Returns next endpoint}
  \KwRet $j$ \;}{}
\end{algorithm}
\end{minipage}
\end{figure*}

\begin{proposition} \label{proposition:betal_gammal} %
Under both $\beta_\ell$-$\ell$-weak $P$-order and $\gamma_\ell$-$\ell$-weak $C$-order assumptions, Alg.~\ref{alg:double_greedy} is $(2 \cdot \max\{1, \beta_\ell, \gamma_\ell\})$-approximate and computes at most $3 \cdot (\ell + 1) \cdot n$ weights.
\end{proposition}


\begin{proof}
The proof follows a different proof schema as the previous ones and we shall this time directly bound each edge of the optimal solution.
First, let us note $M_{opt}$ for the optimal matching of the edges of $E$ and $M$ for the matching obtained by the algorithm.
During the algorithm execution, we refer to any edge $e \in E$ as being ``eliminated'' once the edge cannot be selected in subsequent steps of the algorithm; i.e., the edge $(p_x,c_y)$ is eliminated whenever we have either $\texttt{available}_x = \texttt{False}$ or $\texttt{available}_y = \texttt{False}$. 
Since endpoints are never reconsidered in the algorithm, an eliminated edge stays eliminated till the end of the algorithm.
In the first part of the proof, we show that any edge $e \in M_{opt}$ is eventually eliminated (as any other edge of $E$).
In the second part, we prove the bound on the weight of the produced matching following at its core the same argument behind Algorithm~\ref{alg:matching3}'s bounded approximation (cf. Proposition~\ref{prop_gamma_l}). 
The last part deals with the number of computed weights.

\paragraph{Termination of the algorithm}
Let us show that any edge $e \in E$ is eventually eliminated by the algorithm. 
For the sake of contradiction, suppose there exists an edge $e = (p_x,c_y)$ that is not eliminated by the end of the algorithm and let us consider the first loop iteration when $i = x$. Since the edge is not eliminated, we have $\texttt{available}_i = \texttt{True}$, so the greedy path subroutine is executed with $p_i$ as a starting endpoint.
First, note that the condition set at line~\ref{alg4:greedy_condition} in the subroutine forces the creation of a \textit{path}: cycles are forbidden as any edge sharing an endpoint with a previously considered node of $U$ (hereafter designated as ``backward'' edges) is discarded from being chosen by the procedure.
Now observe that, since the algorithm adds at line~\ref{alg4:optimal_path} all edges belonging to the optimal solution for the path, 
there cannot be an edge of $U$ that is not eliminated (i.e., in the optimal solution all non-selected edges have at least one of its adjacent edges being selected, cf. the proof of Lemma~\ref{lemma:path}).
Thus, $(p_x,c_y)$ must be different from the first edge (in $\sigma_C$ order) of $U$ not to be eliminated at this step, however, note that this step always eliminates at least one edge in $p_x$'s neighborhood.
Hence, the only way for the starting edge not to be eventually eliminated is for the algorithm never to set $\texttt{available}_x = \texttt{False}$ (and thus increase the value for $i$) during the greedy path's subroutine and hence to loop forever on it. However, since one edge in $p_x$'s neighborhood is eliminated every loop iteration, eventually the edge $(p_x,c_y)$ will be part of the greedy path (or no edges are available but this contradicts that $e$ is still available) and will eventually be eliminated, leading to a contradiction.
Observe that this also shows that the algorithm always terminates as for every $1 \leq i \leq \np$, each time $p_i$ is processed by the greedy path subroutine, the size of its ``available neighborhood'' $|N|$ strictly decreases at each while-loop iteration, eventually reaching $|U| = 0$ when the algorithm jumps to the next iteration.

\paragraph{Bounded approximation}

Since any edge $e \in M_{opt}$ is eventually eliminated, let us associate to each edge $e \in M_{opt}$ the while-loop iteration $\texttt{iter}(e)$ where it becomes eliminated. 
Let $M^i_{opt}$ be the set of edges of the optimal matching that are eliminated during the $i$-th loop iteration, i.e. $M^i_{opt} = \{e \in M_{opt} \; | \; \texttt{iter}(e) = i\}$. Also, denote $M_i$ the set of edges that are added to $M$ during the $i$-th iteration, and $U_i$ the value of $U$ at the end of the greedy iterative loop (line~\ref{alg4:selected_edges}) during the same iteration.
By the above claim, the algorithm always terminates, in let us say $r$ loop iterations, so we have  $M_{opt} = \bigcup_{i=1}^{r} M^i_{opt}$ and $M = \bigcup_{i=1}^{r} M_i$.
Now, let us consider the possible reasons for the edges of $M_{opt}$ to be eliminated during a given iteration $I$ (with $i$ being the value of the algorithm's variable $i$ during that iteration):
\begin{enumerate}
    \item[(1)] If the edge $e \in M_{opt}$ is selected by the algorithm at line~\ref{alg4:selected_edges}, it becomes eliminated. Note that in that case, $e \in M$ as well.
    \item[(2)] Suppose the edge $e \in M_{opt}$ is not selected at line~\ref{alg4:selected_edges}. Let us differentiate two subcases:
    \begin{enumerate}
        \item[a.] The edge $e$ was added to the iteratively constructed greedy path $U$. 
        Then, by the above termination arguments, it is also eliminated along all the other non-selected edges forming $U$.
        \item[b.] The edge $e$ was not added to $U$. Here, $e$ is eliminated because it shares an endpoint with one of the selected edges that belong to the path $U$.
    \end{enumerate}
\end{enumerate}

Let us bound the weight of $e$ for each of the possible situations. To ease with the notations, denote $U_I = e_1, \dots, e_k$ the edges forming $U_I$ in the same order as they are added during iteration $I$ (note $U_I$ is possibly empty when all $p_i$'s neighbors are not available at iteration $I$'s start). 
We define the function $f: M^I_{opt} \rightarrow U_I$ projecting the edges of the optimal matching onto the ones of the greedy path.

In Case (1), we set $f(e) = e$ and obviously have $w(e) \leq w(f(e))$. 

In Case (2).a, $e \in U_I$ but $e$ is not selected by Alg.~\ref{alg:double_greedy}; in that case, we also set $f(e) = e$. 

At last let us consider the remaining Case (2).b. Denote $e'$ the first edge of $U_I$ that eliminated $e$ during $I$ (thus, $e'$ shares an endpoint with $e$). 
By definition, $e$ is a forward edge at the step when $e'$ is added to the greedy path as otherwise $e'$ would not be the first edge eliminating $e$.
Let us write $e' = (u,v)$ so that $u$ was the endpoint that was the first parameter of \texttt{next\_edge\_greedy\_path} (also noted $u$ in the algorithm) and $v$ the node returned by the procedure. 
We reason now on the following three cases:
\begin{enumerate}
    \item $e' \cap e = \{u\}$. Set $f(e) = e'$.  
    According to the algorithm, $(u,v)$ has the maximum weight among $N_\ell(u)$, the first $\ell + 1$ available edges of $u$, following the order $\sigma_P$ (resp. $\sigma_C$) of nodes in $P$ (resp. $C$) if $u \in C$ (resp. $u \in P$). 
    Suppose $e$ is among $N_\ell(u)$, then $w(e) \leq w(e')$ as $e'$ has the local maximum weight.
    Suppose $e$ was not among $N_\ell(u)$, then we have:
\begin{itemize}
    \item if $u \in P$, then $v \in C$ and $w(e) \leq \gamma_\ell \cdot w(e'')$ where $e''$ is an edge in $N_\ell(u)$ that is at least $\ell$ edges away from $e'$ in $\sigma_C$ order (the fact that $e''$ exists relies on the same arguments used in the proof of Proposition~\ref{prop_gamma_l}, refer to the proof for further details). Since $w(e'') \leq w(e')$, 
    we get $w(e) \leq \gamma_\ell \cdot w(e')$.
    \item Analogously, we get that if $u \in C$, then $v \in P$ and $w(e) \leq \beta_\ell \cdot w(e')$ by a symmetric argument to the above subcase.
\end{itemize}
    Hence, we have $w(e) \leq \max\{1,\gamma_\ell,\beta_\ell\} \cdot w(f(e))$ regardless if $e$ was in $N_\ell(u)$ or not.

    \item $e' \cap e = \{v\}$ and so that $e' = e_j$ with $1 \leq j \leq k-1$. In other words, $e'$ is not the last edge of $U$.
    Observe that because $e_j \in M$, we have $e_{j+1} \not\in M$.
    We can obtain similar bounds as in the first case but using this time the edge $f(e) = e_{j+1}$, that is the edge that was returned after calling \texttt{next\_edge\_greedy\_path} with first parameter $v$, and since $j < k$ such an edge appears in $U_I$.
    Observe that $e \cap f(e) = \{v\}$ and $e$ cannot be a backward edge as $e'$ is the first edge eliminating $e$, and thus the same arguments as in Case 1 above hold but this time with $e_{j+1}$.
    Hence, we also have $w(e) \leq \max\{1,\gamma_\ell,\beta_\ell\} \cdot w(f(e))$. 
    
    \item $e' \cap e = \{v\}$ and $e' = e_k$, i.e., $e'$ is the last edge of the greedy path $U_I$ and is selected by the algorithm. 
    This case is not possible for two reasons: $e$ is available at that iteration by hypothesis, and $e$ is not a backward edge for $e'$. Hence it is a forward edge in $v$'s neighborhood, but then the greedy path subroutine would have continued to build $U_I$ selecting $e$ on the path if such an edge existed in $v$'s neighborhood.
\end{enumerate}

Hence, for every eliminated edge $e$ of the optimal matching, one can bound its weight in regards to an edge $e_j$ of the greedy path and that is ``responsible'' for the elimination of $e$.
Importantly, $f$ is an \textit{injection} by analyzing the different cases: 
\begin{itemize}
    \item Cases (1) and (2).a. Since $f(e) = e$, there cannot be another optimal edge $e' \in M^I_{opt}$ with $e' \not= e$ also projecting on $e$ as for any edge $x \in M^I_{opt}$, $x$ and $f(x)$ always share an endpoint (in all 3 cases) by construction, forbidding such $e'$ in $M_{opt}$.  
    \item Case (2).b. 
    Suppose there exist two distinct edges $e_1 \in M^I_{opt}$ and $e_2 \in M^I_{opt}$ so that $e' = f(e_1) = f(e_2)$ (obviously $e'$ is distinct from both $e_1$ and $e_2$). 
    Since for every edge $x \in M^I_{opt}$ we have $f(x) \cap x \not= \emptyset$, the three edges $e_1$, $e'$ and $e_2$ must form a path of $3$ edges with $e'$ in the center. 
    Since $f(x)$ is always an edge of $U_I$ in all 3 cases, $e' \in U_I$ and thus $e' = (u,v)$ with $(u,v)$ being a forward edge and $u$ the first node processed by the greedy path subroutine; w.l.o.g. assume $u \in e_1$. 
    By definition of $f$, $e_2$ is projected by $f$ on the next edge of the path so a different edge from $e_2$ if $e_2$ is a forward edge, and if $e_2$ is a backward edge then it is also sent on a previous edge of the path different from $f(e_1)$. 
    Hence, this contradicts $f(e_1) = f(e_2)$.
\end{itemize}

At last, observe that since $M_p$ is chosen optimally among the edges of $U_I$, necessarily we have $w(M_I) \geq \sum_{e \in U_I \setminus M_I} w(e)$ by applying Lemma~\ref{lemma:path} on the path $U_I$.

We are ready to combine the different elements of the proof to obtain a general bound: 
\begin{align*}
    w(M^I_{opt}) &\leq \sum_{e \in M^I_{opt}} w(e) \leq \sum_{e \in M^I_{opt}} \max\{1,\gamma_\ell,\beta_\ell\} \cdot w(f(e)) \\
    &\leq  \max\{1,\gamma_\ell,\beta_\ell\} \left( \sum_{e' \in M_I} w(e') + \sum_{e' \in U_I \setminus M_I} w(e')\right) \\
    &\leq\max\{1,\gamma_\ell,\beta_\ell\} \cdot 2 \cdot w(M_I).  \\
\end{align*}

\vspace{-0.5cm}
Summing over all iterations of the algorithm, we obtain: $$w(M_{opt}) \leq \sum_{i=1}^r w(M^i_{opt})  \leq \sum_{i=1}^r 2 \cdot \max\{1, \beta_\ell, \gamma_\ell\} \cdot w(M_i) \leq  2 \cdot \max\{1, \beta_\ell, \gamma_\ell\} \cdot  w(M).$$

\paragraph{Number of weight calculations}

Let us analyze how many weights are calculated by all iterations. Obviously, if $p_i$ is not available, no further weights are calculated and the algorithm moves to the next iteration. Otherwise, a greedy path is constructed. To do so, for a greedy path of $k$ edges, $(\ell+1) \cdot k$ weights are at most calculated: $\ell +1 $ for each starting endpoint $u$ of each edge $(u,v)$ of the path and $0$ for the last call to \texttt{next\_edge\_greedy\_path} as the path ends when the set $N$ of candidates for the next edge is empty (weights are only calculated when $|N|\geq 2$). 

Over all calculated weights, at least $\lfloor \frac{k}{2} \rfloor$ edges are added to the matching as an optimal matching $M_p$ is always maximal over $U$ (all edges of $U$ have endpoints in $M_p$). 
Doing so makes as many nodes in $P$ and nodes in $C$ non-available and such nodes will not trigger any subsequent weight computation for any of its adjacent edges.
Hence, noting $k_i = |U_i|$ the length (in edges) of the greedy path at the $i$-th iteration (potentially zero), we get that in total the number of weight calculations is at most $(\ell+1) \sum_{i=1}^r k_i$. 

In the worst case (in terms of number of weight calculations), $k_i = 3$ and at each iteration a single edge is added to the matching. Thus, we have $\sum_{i=1}^r k_i \leq 3 \cdot |M|$ and since $|M| \leq n$, we get that at most $3 \cdot (\ell+1) \cdot n = \mathcal{O}(\ell \cdot n)$ weights are ever calculated.

\end{proof}

We would like to highlight that, in Alg.~\ref{alg:double_greedy}, we have chosen to use an identical value for the constant $\ell$ used to explore the local neighborhood in an analogous way whether one deals with a node from $P$ or a node from $C$ (and assuming both weak orders also hold for the same value of $\ell$). This choice simplifies both the algorithm's design and its proof, while reaching the target of a linear number of weight calculations (whenever $\ell$ is constant). The algorithm could be extended to use two different values $\ell_1$ and $\ell_2$ for each set $P$ and $C$, with all arguments being valid when setting $\ell = \max\{\ell_1, \ell_2\}$.

We finally show the lower bound for the presented algorithm.

\begin{figure}
    \centering
    \begin{tikzpicture}[xscale=2,yscale=0.4]
        \node (p1) at (0,0) {$p_1$};
        \node (c0) at (2,-0.55) {$\vdots$};
        \node (c1) at (2,1) {$c_1$};
        \node (c2) at (2,-2) {$c_{\ell+2}$};

        \node (p2) at (4,1) {$p_2$};
        \node (p0) at (4,-0.55) {$\vdots$};
        \node (p3) at (4,-2) {$p_{\ell+3}$};
        
        \draw[-] (p1) -- node[above, near end] {$1 + \varepsilon$} (c1);
        \draw[-] (p1) -- node[above, near end]  {$0.5$} (c0);
        \draw[-] (p1) -- node[below, midway]  {$(1 + \varepsilon) \cdot \gamma_\ell$} (c2);
        
        \draw[-] (c1) -- node[above, near end]  {$1$} (p2);
        \draw[-] (c1) -- node[above, near end]  {$0.5$} (p0);
        \draw[-] (c1) -- node[below, midway]  {$\beta_\ell$} (p3);
    \end{tikzpicture}
    \caption{Example for the lower bound on the approximation ratio of Alg.~\ref{alg:double_greedy}.
    }
    \label{fig:counter_example_betal_gammal}
\end{figure}
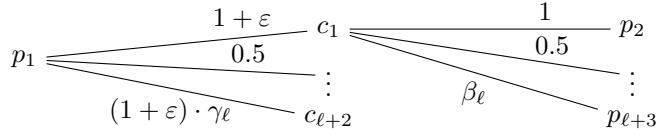

\begin{proposition} \label{example_betal_gammal} 
There exist 
$\beta_\ell > 0$, $\gamma_\ell > 0$ and graph instances and heuristic orders where, under both $\beta_\ell$-$\ell$-weak $P$-order and $\gamma_\ell$-$\ell$-weak $C$-order assumptions, Alg.~\ref{alg:double_greedy} is at best $(2 \cdot \max\{1, \beta_\ell, \gamma_\ell\})/(1+\varepsilon)$-approximate for any $\varepsilon > 0$.
\end{proposition}

\begin{proof}
    We consider the example presented by Figure~\ref{fig:counter_example_betal_gammal}, with the following heuristic orders: $\sigma_P = p_1, \dots, p_{\ell+3}$ and $\sigma_C = c_1, \dots, c_{\ell+2}$. In this example, all edges of the form $(p_1, c_j)$ for $2 \leq j \leq \ell+1$ and $(c_1, p_j)$ for $3 \leq j \leq \ell+2$ have all weight 0.5.
    
    First, let us calculate the weight of the output of Alg.~\ref{alg:double_greedy}. The algorithm starts by considering $p_1$ and builds the greedy path $U = p_1, c_1, p_2$ as those edges $\{(p_1,c_1), (c_1, p_2)\}$ are the maxima in the $(\ell+1)$ local neighborhood of each considered node. 
    Then Alg.~\ref{alg:double_greedy} selects the best matching over $U$ as $(p_1,c_1)$ with weight $1+\varepsilon$ which eliminates all edges of the graph.
    Whenever $\beta_\ell + \gamma_\ell \geq 1+\varepsilon$, the optimal matching is, however, the one formed by selecting $(p_1,c_{\ell+2})$ and $(p_{\ell+3},c_1)$ of weight $\beta_\ell + (1+\varepsilon) \cdot \gamma_\ell$.
    By considering the case $\beta_\ell = (1-\varepsilon) \cdot  \gamma_\ell$ with $\gamma_\ell > 1 + \varepsilon$, the optimal matching has weight $2 \cdot \gamma_\ell = 2 \cdot \max \{\beta_\ell, \gamma_\ell, 1\}$  to compare with $1+\varepsilon$ for Alg.~\ref{alg:double_greedy}'s matching, hence retrieving the desired lower bound for the approximation ratio of Alg.~\ref{alg:double_greedy}.
    
\end{proof}

Let us note that $\varepsilon$ was only introduced in the above proposition so not to have to make the algorithm for choosing the optimal path explicit, and to avoid introducing several examples depending on the chosen matching. If ties were explicitly broken in the calculation of the optimal matching over path (i.e., one always includes the first edge from the path's ``start''), then we can replace in Figure~\ref{fig:counter_example_betal_gammal}'s example $w(p_1,c_1)$ by $1$, $w(p_1,c_{\ell+2})$ by $\gamma_\ell$ and lift $1+\varepsilon$ from Proposition~\ref{example_betal_gammal}.

\subsection{Order Oracles on the Edge Set} \label{sec:edge_orders}

\begin{figure*}[t]
      \centering
\begin{minipage}{0.4\linewidth}
\begin{algorithm}[H]
\footnotesize
\SetKwInOut{Input}{Input}
\SetKwInOut{Output}{Output}
\SetKwInOut{Upon}{Event}
\vspace{0.045cm}
\Input{
A graph $G = (V, E)$ with ordered edge set $E = e_1, \dots, e_m$
}
\Output{$M$, a matching of $E$\;}
    \tcp{Initialization steps}
    $M \leftarrow \emptyset$ \;
    
    \ForEach{$x \in V$}{ 
        $\texttt{available}_x \leftarrow $ True \; 
    }

    \tcp{Process edges one by one}
    \For{$1 \leq i \leq m$}{ 
            $(u, v) \leftarrow e_i$\; 
            \If{$\text{available}_u$ \hbox{and} $\text{available}_v$}
            {
                $M \leftarrow M \cup \{u,v\}$\;
                $\texttt{available}_u$ $\leftarrow$ \texttt{False}\;
                 $\texttt{available}_v$ $\leftarrow$ \texttt{False}\;
            }
        }
\Return $M$\;
\caption{Naive-Edge Matching} \label{alg5:naive_edge} 
\end{algorithm}
\end{minipage}%
\hfill%
\begin{minipage}{0.595\linewidth}
\begin{algorithm}[H]
\footnotesize
    $i \leftarrow 1$\;
    \While{$i \leq m$}{ 
        $(u,v) \leftarrow e_i$\;
        \eIf{$\text{available}_u \land \text{available}_v$}
        {
            $N \leftarrow \{ e_j = (u,v) \in E, i \leq j \leq i+\ell+1 \land  \text{available}_u \land \text{available}_v \}$ \;
            $j \leftarrow 1$ \;
            \If{$|N| > 1$}{
                   $N \leftarrow N[:\ell+1]$ \;
                    \ForEach{$e_j \in N$}{
                       $b_j \leftarrow$ \textsf{weight}$(e_j)$\;
                    }
                    $j \leftarrow \argmax_{j \in N} \; b_{j}$\; 
            }
            $M \leftarrow M \cup \{e_j\}$\;
            $(u,v) \leftarrow e_j$\;
            $\texttt{available}_u, \texttt{available}_v$ $\leftarrow$ \texttt{False}, \texttt{False}\;
        }
        {
            $i \leftarrow i+1$\;
        }
    }
\Return $M$\;
\caption{Local-Edge Matching} \label{alg6:edge_local} 
\end{algorithm}
\end{minipage}
Alg.~\ref{alg6:edge_local} has the same inputs/outputs and initialization steps as in Alg.~\ref{alg5:naive_edge}.
\end{figure*}

We briefly show here that reasoning in terms of edge order is not as interesting as order oracles over nodes. Also, observe that all algorithms presented here do not need to work on bipartite graphs as input (hence solving the more general graph matching problem rather than the assignment problem).

\subsubsection{Strong Edge Order}

\medskip
Consider the following strong ordering assumption on the edges.

\begin{assumption} \label{assumption_zeta}
($\zeta-$strong $E$-order)
Assume $\zeta \ge 0$ and $E$ is ordered by $\sigma_E = e_1, \dots, e_m$, so that for any $e_i,e_j \in E$ with $1 \leq i < j \leq m$, we have $ w(e_j) \leq \zeta \cdot  w(e_i)$.
\end{assumption}

\begin{proposition} \label{proposition_zeta}
   Under $\zeta-$strong $E$-order, Alg.~\ref{alg5:naive_edge} is $(2\cdot \max\{1,\zeta\})$-approximate without calculating any weights.
\end{proposition}

\begin{proof}
    The proof is simple and follows the main structure as the proof of Proposition~\ref{proposition:betal_gammal}. 
    Consider that every selected edge $e \in M$ eliminates up to $2$ edges of the optimal $M_{opt}$, each of weight bounded by $\zeta \cdot w(e)$.
    Hence overall, $w(M_{opt}) \leq 2 \cdot \max\{1,\zeta\} \cdot w(M)$.
    
\end{proof}

It is interesting to note that one can retrieve the approximation bound of $2$ for the classic greedy algorithm from the previous proposition, as it uses an order over the edges (ranked from highest to lowest weight) that guarantees $\zeta \leq 1$.
As previously mentioned, the edge order imposes here to be capable to compare any pair of edges, which is significantly more constraining than any of the other orders studied before (in particular, see Remark~\ref{remark:comparable_edges}).
Also, one may observe that contrary to Alg.~\ref{alg:matching1}, Alg.~\ref{alg:matching2} and Alg.~\ref{alg:matching3} which output an optimal solution whenever strong enough node orders are provided (i.e., when $\beta + \gamma \leq 1$), the introduced algorithm that exploits the strong edge order is shown not to be better than $2$-approximate even when $\zeta < 1$.
In particular, considering for instance the graph of Figure~\ref{fig:counter_example_gamma}, the condition  $\beta + \gamma \leq 1$ cannot be encoded as a strong edge order.

\subsubsection{Weak Edge Order}

\medskip
One can also consider a weak version for the edge order as follows:

\begin{assumption} \label{assumption_zeta_l}
($\zeta_\ell$-$\ell$-weak $E$-order)
Assume $\ell \geq 0$, $\zeta_\ell \geq 0$ and $E$ is ordered by $\sigma_E = e_1, \dots, e_m$, so that for any $e_i,e_j \in E$ with $1 \leq i+\ell < j \leq m$, we have $ w(e_j) \leq \zeta_\ell \cdot  w(e_i)$.
\end{assumption}

\begin{proposition}
   Under $\zeta_\ell$-$\ell$-weak $E$-order, Alg.~\ref{alg6:edge_local} is $(2 \cdot \max\{1,\zeta_\ell\})$-approximate. 
\end{proposition}

\begin{proof}
    Similarly as previously, one can note that selecting a certain edge $e \in M$ at a given iteration $I$ eliminates up to $2$ edges $e', e''$ of the optimal $M_{opt}$.
    Denote $e_i$ the first edge (always available) considered during $I$ and call two edges $e_i, e_j \in E$ close if $|j-i| \leq \ell$.
    The edges $e', e''$ are still available at the iteration that $e$ is added, hence the weight of $e'$ is either (1) smaller than the one of $e$ if $e'$ and $e_i$ are close in $\sigma_E$ order, (2) smaller than $\zeta_\ell$ times the weight of $e$ if $e_i$ and $e'$ are far in $\sigma_E$ order.
    The second case is due to the fact that when $e_i$ is available, we have $w(e') \leq \zeta_\ell \cdot w(e_i) \leq \zeta_\ell \cdot w(e)$.
    The same arguments apply for the weight of $e''$.
    Overall, considering cases (1) and (2), we get $w(M_{opt}) \leq 2 \cdot \max\{1,\zeta_\ell\} \cdot w(M)$.
    
\end{proof}

In Alg.~\ref{alg6:edge_local}, since there cannot be more than $|V|/2$ edges added in total in $M$, and that for all iterations $i$ where $e_i$ is still available an edge is added to $M$, we deduce that there cannot be more than $(\ell+1) \cdot |V|/2$ calls to the weight function. 
As with strong edge order, the imposed order is very restrictive as it forces to be able to compare almost all edges with each other, forbidding only a small constant number of comparisons for each edge. Hence, it is significantly stronger than weak orders on the nodes that only ``locally'' order the edges. 
In the following section, we provide possible instantiations of order oracles and show that edge orders are obtained at the cost of significantly increasing the constants $\zeta$ and $\zeta_\ell$ in the required assumptions.

\subsection{Instantiations of Order Oracles} \label{sec:instantiations}

\begin{figure*}[t]
    \centering
    \footnotesize
\begin{minipage}{.6\textwidth}
  \centering
  \begin{tikzpicture}[yscale=0.5,xscale=0.45]
        \node[inner sep=0cm] (e1) at (-2.2,0) {$I(e_1) = [6.3,11.7]$};
        \node[inner sep=0cm] (e2) at (-2.2,-1) {$I(e_2) = [5.6,10.4]$};
        \node[inner sep=0cm] (e3) at (-2.2,-2) {$I(e_3) = [5.6,10.4]$};
        \node[inner sep=0cm] (e4) at (-2.2,-3) {$I(e_4) = [4.9,9.1]$};
        \node[inner sep=0cm] (e5) at (-2.2,-4) {$I(e_5) = [4.9,9.1]$};
        \node[inner sep=0cm] (e6) at (-2.2,-5) {$I(e_6) = [2.8,5.2]$};
        \node[inner sep=0cm] (e7) at (-2.2,-6) {$I(e_7) = [2.1,3.9]$};
        \node[inner sep=0cm] (e8) at (-2.2,-7) {$I(e_8) = [0.7,1.3]$};

        \draw[->] (-4.5,0) -- (-4.5,-7);
        \node[font=\bfseries] at (-2,1.5) {Ordered Edges};

        \node[font=\bfseries] at (6,1.5) {Linear Representation};

        \node[inner sep=0cm] (e1l) at (6.0,0) {$e_1 [$};\node[inner sep=0cm] (e1r) at (11.7,0) {$]$}; 
        \node[inner sep=0cm] (e2l) at (5.3,-1) {$e_2 [$};\node[inner sep=0cm] (e2r) at (10.4,-1) {$]$}; 
        \node[inner sep=0cm] (e3l) at (5.3,-2) {$e_3 [$};\node[inner sep=0cm] (e3r) at (10.4,-2) {$]$}; 
        \node[inner sep=0cm] (e4l) at (4.6,-3) {$e_4 [$};\node[inner sep=0cm] (e4r) at (9.1,-3) {$]$}; 
        \node[inner sep=0cm] (e5l) at (4.6,-4) {$e_5 [$};\node[inner sep=0cm] (e5r) at (9.1,-4) {$]$}; 
        \node[inner sep=0cm] (e6l) at (2.5,-5) {$e_6 [$};\node[inner sep=0cm] (e6r) at (5.2,-5) {$]$}; 
        \node[inner sep=0cm] (e7l) at (1.8,-6) {$e_7 [$};\node[inner sep=0cm] (e7r) at (3.9,-6) {$]$}; 
        \node[inner sep=0cm] (e8l) at (0.4,-7) {$e_8 [$};\node[inner sep=0cm] (e8r) at (1.3,-7) {$]$}; 
        
        \draw[-] (e1l) -- node[above, midway,font=\tiny,gray] {$(p_1,c_3)$} (e1r);
        \draw[-] (e2l) -- node[above, midway,font=\tiny,gray] {$(p_1,c_2)$} (e2r);
        \draw[-] (e3l) -- node[above, midway,font=\tiny,gray] {$(p_2,c_3)$} (e3r);
        \draw[-] (e4l) -- node[above, midway,font=\tiny,gray] {$(p_1,c_1)$} (e4r);
        \draw[-] (e5l) -- node[above, midway,font=\tiny,gray] { $(p_3,c_4)$} (e5r);
        \draw[-] (e6l) -- node[above, midway,font=\tiny,gray] {$(p_3,c_2)$} (e6r);
        \draw[-] (e7l) -- node[above, midway,font=\tiny,gray] {$(p_2,c_4)$} (e7r);
        \draw[-] (e8l) -- node[above, midway,yshift=1.5pt,font=\tiny,gray] {$(p_2,c_1)$} (e8r);

        \draw[->, gray] (0,-8) -- (12,-8);
        \node[gray] at (0,-8) {$|$}; \node[gray] at (0,-8.66) {$0$};
        \node[gray] at (1,-8) {$|$}; \node[gray] at (1,-8.66) {$1$};
        \node[gray] at (2,-8) {$|$}; \node[gray] at (2,-8.66) {$2$};
        \node[gray] at (3,-8) {$|$}; \node[gray] at (3,-8.66) {$3$};
        \node[gray] at (4,-8) {$|$}; \node[gray] at (4,-8.66) {$4$};
        \node[gray] at (5,-8) {$|$}; \node[gray] at (5,-8.66) {$5$};
        \node[gray] at (6,-8) {$|$}; \node[gray] at (6,-8.66) {$6$};
        \node[gray] at (7,-8) {$|$}; \node[gray] at (7,-8.66) {$7$};
        \node[gray] at (8,-8) {$|$}; \node[gray] at (8,-8.66) {$8$};
        \node[gray] at (9,-8) {$|$}; \node[gray] at (9,-8.66) {$9$};
        \node[gray] at (10,-8) {$|$}; \node[gray] at (10,-8.66) {$10$};
        \node[gray] at (11,-8) {$|$}; \node[gray] at (11,-8.66) {$11$};

    \end{tikzpicture}
     \caption{
     Example of associated intervals to weights following the (yet to be discovered) input of Figure~\ref{fig:double_greedy} and the edges ordered according to the optimistic or centered order; in this example, intervals have been set using original weights $\pm 30$\%.} 
  \label{fig:intervals}
\end{minipage}
\hfill
\begin{minipage}{.35\textwidth}
  \centering
  \begin{tikzpicture}[yscale=0.5,xscale=0.25]
        \node[font=\bfseries] at (15,1.5) {Errors in the $C$-order};

        \node at (8,0) {$e_1 = (p_1,c_3)$};
        \node at (8,-1) {$e_2 = (p_1,c_2)$};
        \node at (8,-2) {$e_4 = (p_1,c_1)$};
        \node at (8,-3) {$e_3 = (p_2,c_3)$};
        \node at (8,-4) {$e_7 = (p_2,c_3)$};
        \node at (8,-5) {$e_8 = (p_2,c_1)$};
        \node at (8,-6) {$e_5 = (p_3,c_4)$};
        \node at (8,-7) {$e_6 = (p_3,c_2)$};

        \draw [decorate,decoration={brace,amplitude=10pt}] (12,0.33) -- (12,-2.33) node[midway, xshift=45pt, text width=2cm]{$I_r(e_2)/I_l(e_1)$ for $p_1$; $\approx$ 1.65};

        \draw [decorate,decoration={brace,amplitude=10pt}] (12,-2.66) -- (12,-5.33) node[midway, xshift=45pt, text width=2cm]{$I_r(e_7)/I_l(e_3)$ for $p_2$; $\approx$ 0.69};

        \draw [decorate,decoration={brace,amplitude=10pt}] (12,-5.66) -- (12,-7.33) node[midway, xshift=45pt, text width=2cm]{$I_r(e_6)/I_l(e_5)$ for $p_3$; $\approx$ 1.06};
        
         \node at (8,-8.66) {};
    \end{tikzpicture}
    \vspace{0.2cm}
  \caption{Bounds on the error associated to each node in $P$ (i.e. according to $C$-order) following the example of Figure~\ref{fig:intervals}. }
  \label{fig:interval_errors}
\end{minipage}
    
\end{figure*}

We propose here how order oracles can be built based on partial or approximate \textit{apriori} knowledge on the weights. As before, our setting is that the exact weights are always accessible but costly to query. Hence, based on knowledge on the particular problem that is being studied, some information can be available to generate a rough estimate of each weight before querying it. We study here a few such instantiations and how strong and weak order oracles can be built upon those approximations.

\paragraph{Strong orders based on interval approximation} 

Suppose that for every edge $e \in E$, there exists an interval $I(e) = [a_e,b_e]$ with $a_e,b_e \in \mathbb{R}^+$ so that one has always the guarantee that $a_e \leq w(e) \leq b_e$. For  $I(e) = [a_e,b_e]$, denote $I_l(e) = a_e$ the interval's lower bound and $I_r(e) = b_e$ the interval's upper bound.
Using such approximation, and without more information on the distribution of weights within their possible intervals for guidance, we can define three possible ordering of the edges:
\begin{enumerate}
    \item ``Optimistic'': order the edges by the highest possible value they can take, i.e., by decreasing value of $I_r(e)$.
    \item ``Centered'': order the edges by the center of their intervals, i.e., by decreasing value of $(I_l(e)+I_r(e))/2$.
    \item ``Pessimistic'': order the edges by the lowest possible value they can take, i.e., by decreasing value of $I_l(e)$.
    
\end{enumerate}

For any given pair of edges $e_1$ and $e_2$, given an order $\sigma_E$ that places $e_1$ before $e_2$, we can associate a maximum possible \textit{relative error} for the order and the given pair as $\hbox{error}(e_1,e_2)_{\sigma_E} = I_r(e_2)/ I_l(e_1)$.
Hence, the order $\sigma_E = e_1, \dots, e_m$ entails a $\zeta$-strong $E$-order (Assumption~\ref{assumption_zeta}) with $\zeta = \max_{e_i,e_j \in E,\, j>i} \hbox{error}(e_i,e_j)_{\sigma_E}$. 
For example, using the intervals of Figure~\ref{fig:intervals}, one can compute a bound of $\zeta = 9.1/4.9 \approx 1.85$ achieved by the pair $(e_4, e_5)$.
In general, one can bound the maximum error of using such heuristic order. Since in the worst case the maximum overlap includes a full interval, the maximum error between two edges is bounded by the maximum over all $e \in E$ of $ I_r(e)/I_l(e) \leq (I_l(e) + I_{\max})/I_l(e)$ with $I_{\max}$ the maximum length of an interval according to $I$. 
Since $ (I_l(e) + I_{\max})/I_l(e)$ is a decreasing function in $I_l(e)$, we have $(I_l(e) + I_{\max})/I_l(e) \leq (w_{\min}+I_{\max})/w_{\min}$ where $w_{\min}$ is the minimum value possible for the weights, {\em i.e.}, the minimum of $I_l(e)$. This gives a bound of e.g.~$\zeta = I_{\max}+1$ if $w_{\min} = 1$.

Now, assume we define the following orders $\sigma_P$ and $\sigma_C$ on $P$ and $C$: number vertices as they appear in $\sigma_E$.
For example, using the intervals of Figure~\ref{fig:intervals}, $\sigma_P = p_1, p_2, p_3$ and $\sigma_C = c_3, c_2, c_1, c_4$. 
Using such node order, one can compute similarly bounds on the maximum ``error'' that the order can imply, but this time only edges sharing an endpoint do produce errors (in $P$ if considering $\sigma_P$, and in $C$ if considering $\sigma_C$).
Using the same example as previously, we obtain a bound of $\beta = \gamma \approx 1.65$ because of the pairs $\{(p_1,c_3),(p_1,c_2)\}$ for $\gamma$ and $\{(p_1,c_3),(p_2,c_3)\}$ for $\beta$.
In particular, this example illustrates well that based on the exact same approximation information (interval weights), order oracles on nodes may entail tighter bounds for the discovery algorithms than using an edge order based on the same interval weight approximation.

\paragraph{Weak orders based on interval approximation} 

Define the \textit{overlap count} (OC) as the maximum number of overlapping edges for any given edge, i.e., OC is calculated as
$$OC(E,I) \,=\, max_{e \in E} \; |\{e' \;\mid\; e' \in E \:\land\: e'\neq e \:\land\: I_{l}(e') < I_{r}(e) \:\land\: I_{r}(e') > I_{l}(e)\}|.$$
For a given set of interval weights, we have $\zeta_\ell$-$\ell$-weak $E$-order holding with $\zeta_\ell \leq 1$ whenever $\ell = OC(E,I)$. 
Under all aforedefined edge orders $\sigma_E$ that are based on intervals (centered, pessimistic or optimistic)\footnote{The only condition required on the order is that two non-overlapping edges are placed in the order of their endpoints.}, we have the property that for any given edge $e$ in $\sigma_E$, all edges appearing beyond the overlap count will have at most a strictly smaller weight than the lowest possible value for $e$.
Observe that the same property holds for weak node orders: we have $\beta_\ell$-$\ell$-weak $P$-order holding with $\beta_\ell \leq 1$ whenever $\ell = OC_C(E,I)$. Here, $OC_C(E,I) \leq OC(E,I)$ is the overlap count only accounting for edges sharing the same endpoint in $C$ (that is the maximum overlap count for any $c$).
Symmetrically, we have $\gamma_\ell$-$\ell$-weak $C$-order holding with $\gamma_\ell \leq 1$ when $\ell = OC_P(E,I)$.
For example, using the intervals of Figure~\ref{fig:intervals} and $\ell = 1$, we obtain $\zeta_1 \approx 1.65$, $\zeta_2 \approx 1.62$, $\zeta_3 \approx 1.44$, $\zeta_4 \approx 0.82$ etc; whereas $\beta_1 = 0$, $\gamma_1 \approx 1.44$ and $\gamma_2 = 0$.
That means, assuming an overlap count of $\ell$ for both $P$ and $C$ weak node orders, only $\mathcal{O}(\minpc \cdot \ell)$ edges are examined with Alg.~\ref{alg:double_greedy} to guarantee a $2$-approximation for the produced matching (i.e., we have $\beta_\ell,\gamma_\ell \leq 1$ then, cf.~Proposition~\ref{proposition:betal_gammal}). 
 That is, whenever we have $m = |E| = \Omega(\minpc{}^2)$ and interval weights with a constant overlap count, only $\mathcal{O}(\sqrt{m})$ edges have to be queried to generate a $2$-approximate matching.

\paragraph{Function approximation}

Suppose there exists a function $f$ that provides an approximation of the weighting function $w$, such as
$ w(e) - \Delta \leq f(e) \leq w(e) + \Delta$ or $ (1-\varepsilon) \cdot w(e) \leq f(e) \leq (1+\varepsilon) \cdot w(e)$ for some $\Delta \geq 0$ (resp. $\varepsilon \geq 0$).
In this case, one can set $I(e) = [f(e)-\Delta,f(e)+\Delta]$ (absolute error guarantee) or $I(e) = [f(e)\cdot (1-\varepsilon),f(e) \cdot (1+\varepsilon)]$ (relative error guarantee) so that all previous results on interval weights hold on approximate weights, including the implications based on the value of the overlap count.
For example, in the relative error guarantee, by the precedent arguments, the maximum error based on the optimistic order is bounded by $(1+\varepsilon)/(1-\varepsilon)$, e.g. the error is less than $2$ if $\varepsilon \leq 1/3$.
This entails that the \naiveedge{} Matching algorithm is $2$-approximate without calculating any weights whenever the heuristic order of the edges is so that each approximated weight is within $77\%-133\%$ of its real value (using Proposition~\ref{proposition_zeta}).
In addition, depending on how much overlap there is in the estimations, tighter approximation ratios can be obtained by calculating the exact value of some of the weights using our discovery algorithms.
In the best scenario, the local approximations do not overlap at all (at least when considering only the edges in the neighborhood of each node) and 
Alg.~\ref{alg:matching2} outputs a $2$-approximation of the optimal without calculating any weights.

\section{Extensions to One-to-Many Assignment Problems and Applications} \label{sec:hypergraph_matchings}


We extend our results in this section to one-to-many assignment problems, i.e., when each member of the set $P$ can be matched with up to $k\geq 2$ different members of the set $C$ instead of only one in the previously studied one-to-one assignment problem.
Assignments where $p \in P$ may be allowed to be paired with up to $k$ elements of $C$ can take two forms: either the assignment of $p$ to many tasks follow the same weighting function as in the one-to-one assignment, or the weight of assigning $p$ to a subset $X$ of tasks is different from the sum of assigning $p$ to the individual tasks from $X$, that is $w(\{p\} \cup X) \not= \sum_{c \in X} w(p,c)$, extending the weighting function to subsets of $P \cup C$. In the former case that we hereafter refer to as \textit{simple one-to-many assignment problem}, the problem is a straightforward generalization of the one-to-one assignment problem, whereas in the latter case, the problem corresponds to a form of bipartite hypergraph matching, a significantly more challenging problem.

\subsection{Simple One-to-many Assignment Problem} 
\label{subsec:simple_otmap}

One can reduce any simple one-to-many assignment problem to a one-to-one assignment problem in the following manner.
For each $p \in P$, make $k$ copies $p^1, \dots, p^{k}$ of node $p$, while keeping the original weights, \ie{} $\forall c \in C, w(p^j, c) = w(p, c)$.
Finally, solve the maximum matching in bipartite graphs problem with the input $P'_k = \{ p^j \;|\; j \in [1..k], p \in P\}$ and~$C$. 
This reduction allows to extend all our results to simple one-to-many assignment problems with all bounds shown in Table~\ref{tab:results} still holding the same way over $G = (P'_k \cup C, E)$ using the algorithms we have introduced in this work (upon using appropriate order oracles). 
In detail, we can re-use all the discovery algorithms that have been introduced so far and adapt if needed how the original order oracles translate to this situation.  
That is, assuming order oracles $\sigma_P$ and $\sigma_C$ are available for $P$ and $C$, one has to extend $\sigma_P$ to $\sigma_{P'_k}$. 
Let us consider three intuitive strategies: (1) \textsc{Round Robin} places all vertices in $P$ first $\{p_i^1 \;|\; p_i \in P\}$ ordering them in $\sigma_P$, then cycle $k-1$ more times over the other copies of $P$ according to $\sigma_P$ each time; (2) \textsc{Single Pass} places first all copies of $p_1$ before moving to all copies of $p_2$, etc; (3) \textsc{Classic Greedy} orders all edges (including the copies) in the usual decreasing order of edge weights.
Here, \textsc{Single Pass} preserves the assumption bounds associated with $\sigma_P$.
This is because pairs of edges involving the ``new edges'' are either involving a node in $P$ and one of its copies (simply enforcing $\beta \geq 1$) or 
two distinct nodes in $P$ but appearing in $\sigma_{P'_k}$ in the same order as in $\sigma_P$.
Since \textsc{Round Robin} shuffles how nodes appear in $\sigma_{P'_k}$, it does not preserve the original assumption bounds.
At last, \textsc{Classic Greedy} provides a $\zeta$-strong edge order albeit enforcing $\zeta = 1$ when $k \geq 2$ because of the copies.

\subsection{General One-to-many Assignment Problem} 
\label{one-to-many-assignments:hypergraphs}

Following our original motivations stemming from energy systems~\cite{duvignau2024geographical}, we further explain here how to extend our results to the bipartite hypergraph matching problem. 

\paragraph{Bipartite hypergraph matching problem}

Contrary to usual graphs where the definition of bipartiteness is rather intuitive and unique, the notion accepts several variants for its hypergraph equivalent~\cite{lovasz2009matching}. One natural extension of bipartition to hypergraphs assumes for the vertex set $V$ of the hypergraph $\mathcal{G}$ to be partitioned into two disjoint sets $P$ and $C$, 
such that every hyperedge of $\mathcal{G}$ contains at least one vertex from $P$ and one from $C$. 
To match the setting of a \textit{general one-to-many} assignment problem, we rather restrict hyperedges to have exactly one vertex from $P$ (and since hyperedges contain at least two vertices, this is a subset of the bipartite hypergraphs in the wider definition).
The $k$-\textit{BHM-Discovery} problem seeks, given a bipartite hypergraph $\mathcal{G} = (P \cup C, \mathcal{E})$, to find the maximum-weight matching of the hyperedges of $\mathcal{E}$ where hyperedges are of size at most $k$. As with the previous discovery problems, hyperedge weights are not given as input and must be individually queried.
For $k \geq 4$ and if no weight assumptions are provided, the $k$-BHM-Discovery problem is not approximable within a factor of $o(k/\log k)$ in polynomial time, unless P = NP (cf.~\cite{duvignau2024geographical}, based on a reduction to the $k$-bounded hypergraph matching problem~\cite{HazanSS06}).
By calculating the weight of all $\mathcal{O}(N^{k})$ possible hyperedges with $N = \max\{|P|,|C|\}$, the best approximation algorithms~\cite{berman2000d, Neuwohner23} achieve slightly less than a $(k+1)/2$ approximation ratio.


\paragraph{Approximation bounds for the Peer-To-Peer Energy Sharing application}

As with simple one-to-many assignments, we can re-use all the discovery algorithms that have been introduced in this work by simply running them on the input $G = (P'_{k-1} \cup C, E)$, with $E = P'_{k-1} \times C$, assuming we can calculate pairwise weights $w(p,c)$ that provide indications for the weight of hyperedges containing both $p$ and $c$, and $w(p,c) = 0$ if $p$ and $c$ do not appear in any hyperedges of the original input $\mathcal{G} = (P \cup C, \mathcal{E})$. 
The output hypergraph matching is then obtained by merging together the different copies so to create groups of size up to $k$.
In this general setting, our results do not carry over because the weight $w(\{p\} \cup X)$ of a hyperedge $\{p\} \cup X \in \mathcal{E}$ is different from the sum of the individual pairwise weights, i.e., $\sum_{c \in X} w(p,c)$.
However, it is shown in~\cite{duvignau2024geographical} that if for any hyperedge $e = \{p\} \cup X$ of the input we have $\alpha_1(k) \leq w(e)/w'(e) \leq \alpha_2(k)$ where $w'(e) = \sum_{c \in X} w(p,c)$, then an $r$-approximate matching algorithm entails an algorithm with an $r\cdot \alpha_2(k)/\alpha_1(k)$ approximation ratio for the $k$-BHM-Discovery problem.
For the practical application of ``Peer-To-Peer Energy Sharing'' considered in~\cite{duvignau2024geographical}, bounds of $\alpha_1(k) = \frac{1}{k-1}$ and $\alpha_2(k) = 1$ are proven. 
This thus entails discovery algorithms of approximation ratio $(k-1) \cdot \varepsilon$ where $\varepsilon$ corresponds to the bound as shown in Table~\ref{tab:results}, depending on the chosen matching algorithm and strength of the involved order oracles. 
In this application, the clear advantage of using discovery algorithms for the bipartite hypergraph matching problem instead of one based on exhaustively enumerating all hyperedges resides in calculating at most $\mathcal{O}(n)$ weights (e.g. using Alg.~\ref{alg:matching3} or Alg.~\ref{alg:double_greedy} with a constant value for $\ell$) instead of $|\mathcal{E}| = \mathcal{O}(s \cdot q^{k-1})$ where $s = |P|$, $q = |C|$, $n = \min\{s,q\}$.
Our results thus entail that those efficient greedy algorithms also provide proven approximation guarantees depending only on the quality of the heuristic orders used to process the input.

\section{Conclusions} \label{sec:conclusions}

We have in this work extended the notion of discovery algorithms to assignment problems, and we believe our work provides useful theoretical bounds for algorithms that are efficient in practice.
The algorithms that we have developed require only weaker assumptions on processing orders for the nodes than a total ordering of all edges and achieve a bounded approximation ratio depending only on the quality of the heuristic orders used.
Furthermore, to provide a bounded-approximation solution to practical applications, we also discuss here extensions of the greedy algorithms introduced earlier to one-to-many assignments and further study their performances based on assumptions stemming from real-world data.
We note that for a given input, processing order and with access to all weights, one can 
compute efficiently the exact values of the heuristic parameters which can become good estimations for bounds on larger instances in a given application.
Our findings open up for further rehabilitation of greedy algorithms in theoretical analysis, and advocate that greedy algorithms do not only often provide computationally-efficient solutions to hard problems but can also be formally analyzed within the scope of concrete applications.


\section*{Acknowledgements}

Romaric Duvignau is partially supported by the TANDEM project within the framework of the Swedish Electricity Storage and Balancing Centre (SESBC), funded by the Swedish Energy Agency together with five academic and twenty-six non-academic partners.
Ralf Klasing is partially supported by the ANR project TEMPOGRAL (ANR-22-CE48-0001).



\bibliographystyle{plain}
\bibliography{main}




\end{document}